\documentclass[Times, 10pt,twocolumn]{article} 
\usepackage{latex8}
\usepackage{times}
\usepackage{verbatim} 
\usepackage{cite}
\usepackage{clrscode3e}
\usepackage[pdftex]{graphicx}
\usepackage{setspace}
\usepackage[cmex10]{amsmath}
\usepackage{amsmath, amsthm, amssymb}
\usepackage{algorithmic}
\usepackage{algorithm}
\usepackage{float}

\usepackage{subfig}
\usepackage{graphicx}

\usepackage{caption}
\usepackage[compact]{titlesec}
\usepackage{fullpage}
\newtheorem{theorem}{\textbf{Theorem}}

\newtheorem{lemma}{\textbf{Lemma}}

\theoremstyle{definition}
\newtheorem{definition}{\textbf{Definition}}

\begin{document}  

\title{\protect\vspace{-.5in}{Supporting Soft Real-Time Sporadic Task Systems on Heterogeneous Multiprocessors with No Utilization Loss}\thanks{This work was supported by a start-up
grant from the University of Texas at Dallas.}}
\author{Guangmo Tong and Cong Liu\\
Department of Computer Science, University of Texas at Dallas}

\maketitle
\thispagestyle{empty}

\begin{abstract}

Heterogeneous multicore architectures are becoming increasingly popular due to their potential of achieving high performance and energy efficiency compared to the homogeneous multicore architectures. In such systems, the real-time scheduling problem becomes more challenging in that processors have different speeds. A job executing on a processor with speed $x$ for $t$ time units completes $(x \cdot t)$ units of execution. Prior research on heterogeneous multiprocessor real-time scheduling has focused on hard real-time systems, where, significant processing capacity may have to be sacrificed in the worst-case to ensure that all deadlines are met. As meeting hard deadlines is overkill for many soft real-time systems in practice, this paper shows that on soft real-time heterogeneous multiprocessors, bounded response times can be ensured for globally-scheduled sporadic task systems with no utilization loss. A GEDF-based scheduling algorithm, namely GEDF-H, is presented and response time bounds are established under both preemptive and non-preemptive GEDF-H scheduling. Extensive experiments show that  the magnitude of the derived response time bound is reasonable, often smaller than three task periods. To the best of our knowledge, this paper is the first to show that soft real-time sporadic task systems can be supported on heterogeneous multiprocessors without utilization loss, and with reasonable predicted response time.

\end{abstract}

\section{Introduction}
\label{sec:intro}

Given the need to achieve higher performance without driving up power consumption and heat dissipation, most chip manufacturers have shifted to multicore architectures. An important subcategory of such architectures are those that are heterogeneous in design. By integrating processors with different speeds, such architectures can provide high performance and power efficiency \cite{liu2012power}. Heterogeneous multicore architectures have been widely adopted in various computing domains, ranging from embedded systems to high performance computing systems. 

Most prior work on supporting real-time workloads on such heterogeneous multiprocessors has focused on hard real-time (HRT) systems. Unfortunately, if all task deadlines must be viewed as hard, significant processing capacity must be sacrificed in the worst-case, due to either inherent schedulability-related utilization loss---which is unavoidable under most scheduling schemes---or high runtime overheads---which typically arise in optimal schemes that avoid schedulability-related loss.\footnote{Such utilization loss may exist even in a homogeneous HRT multiprocessor system where all processors have the same speed \cite{baruah2007techniques, shin2012rtss, davis2011survey, bertogna2007response, davis2012optimal}.} In many systems where less stringent notions of real-time correctness suffice, such loss can be avoided by viewing deadlines as soft. In this paper, we consider the problem of scheduling soft real-time (SRT) sporadic task systems on a heterogeneous multiprocessor; the notion of SRT correctness we consider is that response time is bounded.

All multiprocessor scheduling algorithms follow either a \textit{partitioning} or \textit{globally-scheduling approach} (or some combination of the two). Under partitioning, tasks are statically mapped to processors, while under global scheduling, they may migrate. Under partitioning schemes, constraints on overall utilization are required to ensure timeliness even for SRT systems due to bin-packing-related loss. On the other hand, a variety of global schedulers including the widely studied global earliest-deadline-first (GEDF) scheduling algorithm are capable of ensuring bounded response times for sporadic task systems on a homogeneous multiprocessor, as long as the system is not over-utilized \cite{Devi}. Motivated by this optimal result, we investigate whether GEDF remains optimal in a heterogeneous multiprocessor SRT system.

\paragraph{Key observation.} Under GEDF, we select $m$ highest-priority jobs at any time instant and execute them on $m$ processors. The job prioritization rule is according to earliest-deadline-first. Regarding the processor selection rule (i.e., which processor should be selected for executing which job), it is typical to select processors in an arbitrary manner. On a homogeneous multiprocessor, such an arbitrary processor selection rule is reasonable since all processors have identical speeds. However, on a heterogeneous multiprocessor, this arbitrary strategy may fail to schedule a SRT sporadic task system that is actually feasible under GEDF. Consider a task system with two sporadic tasks $\tau_1(2,2)$ and $\tau_2(4,2)$ (notation $\tau_i(e_i, p_i)$ denotes that task $\tau_i$ has an execution cost of $e_i$ and a period of $p_i$) scheduled on a heterogeneous multiprocessor with two processors, $M_1$ with speed of one unit execution per unit time and $M_2$ with speed of two units execution per unit time. Assume in the example that task deadlines equal their periods and priority ties are broken in favor of $\tau_1$. Fig.\ref{fig:simpleexample}(a) shows the corresponding GEDF schedule with an  arbitrary processor selection strategy for this task system. As seen in the figure, if we arbitrarily select processors for job executions, the response time of $\tau_2$ grows unboundedly. However, if we define specific processor selection rules, for example always executing tasks with higher utilizations on processors with higher speeds, then this task system becomes schedulable as illustrated in Fig.\ref{fig:simpleexample}(b).

\begin{figure}[t]
	\begin{center}
	\includegraphics[width=2.8in]{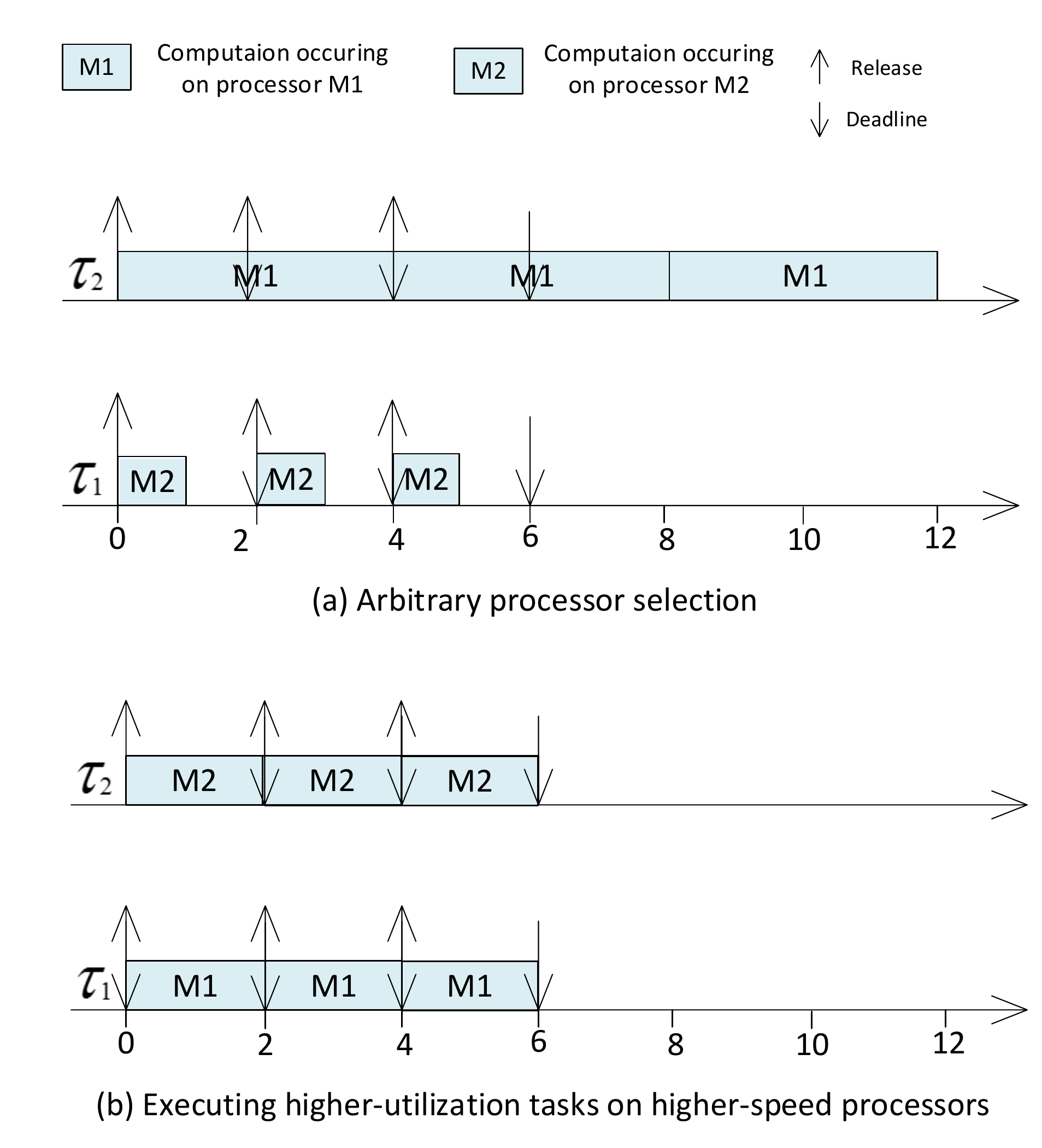} 
	\end{center}  
\vspace{-3mm}
\caption{Motivational example.}
\vspace{-3mm}
\label{fig:simpleexample}
\end{figure}

%?Guangmo: explain what happens in this schedule. Then talk about Fig.~?(b) where a GEDF schedule but with a specific processor selection rule is defined (i.e., tasks with heavier execution requirements are assigned to faster cores). Actually the task system becomes schedulable. Intuitively discuss why (because of the different speeds).?

The above example suggests that \textit{on a heterogeneous multiprocessor, GEDF's processor selection strategy is critical to ensuring schedulability}. Motivated by this key observation, we consider in this paper whether it is possible to develop a GEDF-based scheduling algorithm with a specific processor selection rule, which can schedule SRT sporadic task systems on a heterogeneous multiprocessor with no utilization loss.

\paragraph{Overview of related work.} The real-time scheduling problem on heterogeneous multiprocessors has received much attention \cite{baruah2011partitioned, kumar2004single, andersson2010assigning, baruah2005partitioned, funk2005task, funk2001line, liu2012power}. Most such work has focused on HRT systems, which inevitably incur utilization loss. Partitioning approaches have been proposed in \cite{baruah2004feasibility, baruah2011partitioned, baruah2005partitioned, funk2005task, funk2001line, liu2012power, andersson2010assigning} and quantitative approximation ratios have been derived for quantifying the quality of these approaches. Unfortunately, such partitioning approaches inherently suffer from bin-packing-related utilization loss, which may be significant in many cases. The feasibility problem of globally scheduling HRT sporadic task systems on a heterogeneous multiprocessor has also been studied \cite{baruah2004feasibility}. In \cite{cong2012energy}, a global scheduling algorithm has been implemented on Intel's QuickIA heterogeneous prototype platform and experimental studies showed that this approach is effective in improving the system energy efficiency.

The SRT scheduling problem on a heterogeneous multiprocessor has also been studied  \cite{leontyev2007tardiness}. A semi-partitioned approach has been proposed in \cite{leontyev2007tardiness}, where tasks are categorized as either ``fixed'' or ``intergroup'' and processors are partitioned into groups according to their speeds. Tasks belonging to the fixed category are only allowed to migrate among processors within in the task's assigned group. Only tasks belonging to the migrating category are allowed to migrate among groups. Although this approach is quite effective in many cases, it yields utilization loss and requires several restricted assumptions (e.g., the system contains at least 4 processors and each processor group contains at least two processors). Different from this work, our focus in this paper is on designing GEDF-based global schedulers that ensure no utilization loss under both preemptive and non-preemptive scheduling.

%summarize prior work on HRT and SRT scheduling on heterogeneous multiprocessors. Most of such work focus on partitioning approaches (the only global scheduling paper is due to Sanjoy's Baruah's RTSS 2004 paper). Talk about our different objectives.

\paragraph{Contribution.} 
In this paper, we design and analyze a GEDF-based scheduling algorithm GEDF-H (GEDF for Heterogeneous multiprocessors) for supporting SRT sporadic task systems on a heterogeneous multiprocessor that contains processors with different speeds. The derived schedulability test shows that any sporadic task system is schedulable under both preemptive and non-preemptive GEDF-H scheduling with bounded response times if $U_{sum} \leq R_{sum}$ and Eq.(1) hold, where $U_{sum}$ is the total task utilization, $R_{sum}$ is the total system capacity, and Eq.(1) is an enforced requirement on the relationship between task parameters and processor parameters. We show via a counterexample that task systems that violate Eq.(1) may have unbounded response time under any scheduling algorithm. As demonstrated by experiments, the response time bound achieved under GEDF-H is reasonably low, often within three task periods. Thus, GEDF-H is able to guarantee schedulability with no utilization loss while providing low predicted response time.

\paragraph{Organization.} This paper is organized as follows. In Sec.2, we describe the system model. Then in Sec.3, we describe GEDF-H. In Sec.4, we present our schedulability analysis for GEDF-H and derive the resulting schedulability test. In Sec.5, we show experimental results. We conclude in Sec.6

\section{System Model}
\label{sec:model}

In this paper, we consider the problem of scheduling $n$ sporadic SRT tasks on $m \geq 1$ heterogeneous processors. Let set $\tau=\lbrace \tau_1,...,\tau_n \rbrace$ denote the $n$ independent sporadic tasks and $\chi$ denotes the set of $m$ heterogeneous processors.

Assume there are $z\geq 1$ kinds of processors distinguished by their speeds. Let $\chi_{i}$ $( 1\leq i\leq z)$ and $M_{i}\geq 1$ denote the subset of the $i$th kind of processors in $\chi$ and the number of processors in $\chi_{i}$ respectively. Thus, $\chi=\cup_{i=1}^{z}\chi_{i}$ and $m=\sum_{i=1}^{z}M_{i}$. We assume the processors in $\chi_{1}$ have unit speed and processors in $\chi_{i}$ have speed $\alpha_{i}$ (i.e.$, \alpha_{1}=1, \alpha_{i} < \alpha_{i+1}$). For clarity, we use $\alpha_{max}$ to denote the maximum speed (i.e., $\alpha_{max}=\alpha_z$). Let $R_{sum}= \sum_{i=1}^{z}\alpha_{i} \cdot M_{i}$. 

We define the unit workload to be the amount of work done under the unit speed within a unit time. We assume that each job of $\tau_i$ executes for at most $e_i$ workload which needs $e_i$ time units under the unit speed. The $j^{th}$ job of $\tau_i$, denoted $\tau_{i,j}$, is released at time $r_{i,j}$ and has an absolute deadline at time $d_{i,j}$.  Each task $\tau_i$ has a period $p_i$, which specifies the minimum time between two consecutive job releases of $\tau_i$, and a deadline $d_i$, which specifies the relative deadline of each such job, i.e., $d_{i,j}=r_{i,j}+d_i$. The utilization of a task $\tau_i$ is defined as $u_i=e_i/p_i$, and the utilization of the task system $\tau$ as $U_{sum}=\sum_{\tau_i \in \tau} u_i$. An sporadic task system $\tau$ is said to be an \textit{implicit-deadline} system if $d_i = p_i$ holds for each $\tau_i$. Due to space limitation, we limit attention to implicit-deadline sporadic task systems in this paper.

%\small\footnote{$\tau$ is said to be a \textit{constrained-deadline} system if, for each task $\tau_i \in \tau$, $d_i \leq p_i$, and an \textit{arbitrary-deadline} system if, for each $\tau_i$, the relation between $d_i$ and $p_i$ is not constrained (e.g., $d_i > p_i$ is possible).} \normalsize 

\normalsize
%A common case for real-time workloads is that both self-suspending tasks and computational tasks (which do not suspend) co-exist. To reflect this, we let $U_{sum}^{s}$ denote the total utilization of all self-suspending tasks, and $U_{sum}^{c}$ denote the total utilization of all computational tasks. 

Successive jobs of the same task are required to execute in sequence. If a job $\tau_{i,j}$ completes at time $t$, then its \textit{response time} is $max(0, t-r_{i,j})$. A task's response time is the maximum response time of any of its jobs. Note that, when a job of a task misses its deadline, the release time of the next job of that task is not altered. We require $u_i \leq \alpha_{max}$, and $U_{sum} \leq R_{sum}$, for otherwise the response time may grow unboundedly. 
%For simplicity, we henceforth assume that each job of any task $\tau_i$ executes \textit{exactly} $e_i$ workload. 

%%In this paper, we assume that time is integral, and for any task $\tau_i \in \tau$, each of $e_i \geq 1$, $p_i \geq 1$, and $d_i \geq 1$  is a non-negative integer. If a job executes at time instant $t$ then it executes during the entire time interval $[t,t+1)$. Correspondingly, we assume $\alpha_{i}$ are positive integers.%Similarly, if a processor is idle or busy at $t$, then this processor is idle or busy within $[t,t+1)$. 

%So that our analysis can be more accurately applied in settings where a task's total suspension time varies from job to job, we assume that a fixed parameter $H$ ($H \geq 1$) is specified and that $S_i^{H}$ denote the maximum total self-suspension length for any $H$ ($H \geq 1$) consecutive jobs of task $\tau_i$. Note that if $H=1$, then a maximum per-job total suspension length is being assumed. 

Under GEDF, released jobs are prioritized by their absolute deadlines. We assume that ties are broken by task ID (lower IDs are favored). Thus, two jobs cannot have the same priority. In this paper, we use continuous time system and parameters are positive rational numbers.  

On a heterogeneous multiprocessor, the response time can still grow unboundedly, even if $u_i \leq \alpha_{max}$ and $U_{sum} \leq R_{sum}$ hold. This is illustrated by the following counterexample.   

\begin{figure}[t]
	\begin{center}
	\includegraphics[width=3in]{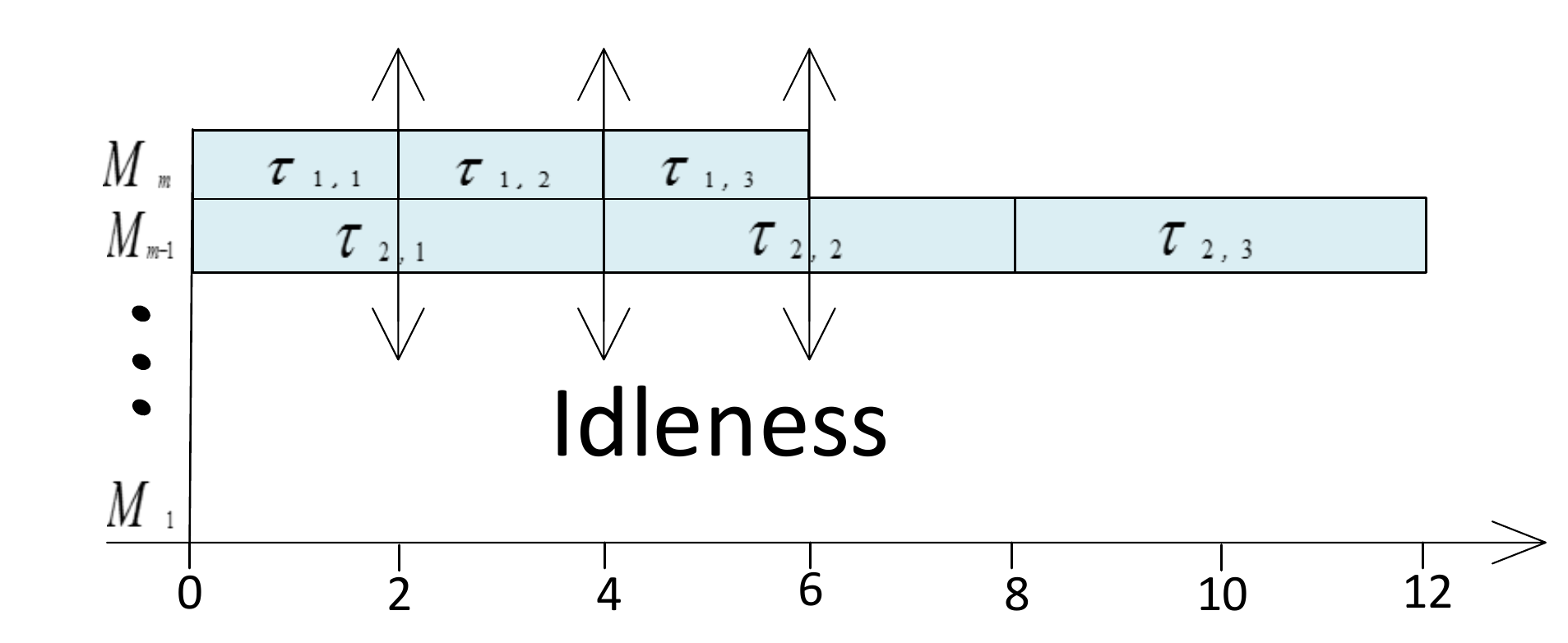} 
	\end{center} 
\vspace{-5mm}
\caption{\small GEDF schedule of the tasks in counterexample.}
\vspace{-5mm}
\label{fig:counter}
\end{figure}

\paragraph{Counterexample.} Consider a sporadic task system with two tasks $\tau_1= \tau_2= ( 2, 1)$ and a heterogeneous multiprocessor with $m\geq 3$ processors where $M_m$ has a speed of $\alpha_{max}=2$ and other $m-1$ processors have unit speed. For this system, $u_1= u_2= \alpha_{max}=2$ and $R_{sum}=2+(m-1)=m+1\geq 4=U_{sum}$. The ratio of $U_{sum}/R_{sum}$ may approximate to 0 when $m$ is arbitrarily large. However, as seen in the GEDF schedule illustrated in Fig.\ref{fig:counter}, regardless of the value we choose for $m$, the response time of $\tau_2$ still grows unboundedly. Actually, we analytically prove that this task system cannot be scheduled under any global or partitioned schedule algorithm. This counterexample implies that a task system may not be feasible on a heterogeneous multiprocessor even provided $U_{sum} \leq R_{sum}$. As seen in Fig.\ref{fig:counter}, adding more unit speed processors  does not help either because there are two tasks with utilization greater than $1$ while only one processor with speed greater than $1$. Motivated by this observation, we enforce the following requirement.

Let $\varPhi_i= \{ \tau_j| \alpha_i<  u_j\}$, $1\leq i< z$, and $|\varPhi_i|$ be the number tasks in $\varPhi_i$. Let $\varPhi_0 = \tau$. Let $\Psi_i=\bigcup_{j=i+1}^{z}\chi_j$, $0\leq i< z$,  and $|\Psi_i|$ be the number of processor in $\Psi_i$. Thus, $\Phi_i$ is the set of tasks that would fail their deadlines if run entirely on a processor of type i or lower, and $\Psi_i$ is the set of processors of type i+1 or higher. For each $1\leq i< z$, we require
\begin{equation}
\label{eq:Restriction}
|\varPhi_i|\leq |\Psi_i|
\end{equation}
Intuitively, Eq.(\ref{eq:Restriction}) requires that if we have $k$ processors with speed $> \alpha_i$, then at most $k$ tasks with utilization $> \alpha_i$
can be supported in the system, which is also a reasonable requirement in practice. Note that, other than $U_{sum}\leq R_{sum}$, we do not place any restriction on $U_{sum}$.

\paragraph{Example 1.} Consider a task system with 4 tasks ,$\tau_1=(2, 1), \tau_2=(2, 1),  \tau_3=(1, 1), \tau_4=(1, 1)$ and a heterogeneous multiprocessor consisting of 3 processors with 2 kinds of speeds where $\alpha_1= 1$, $\alpha_2= 2.5$. For this task system, $u_1= u_2= 2, u_3= u_4= 1$ and we have $\chi_1=\{M_1\}$, $\chi_2=\{M_2, M_3\}$,  $\varPhi_0=\{\tau_1, \tau_2, \tau_3, \tau_4\}$, $\varPhi_1=\{\tau_1, \tau_2\}$, and $\Psi_1= \{ M_2, M_3\}$. Thus, we have $|\varPhi_1|= 2 \leq |\Psi_1|= 2$. This system clearly meets the requirement stated in Eq.(\ref{eq:Restriction}). %And it is easy to check that the system model in Counterexample 1 does not meet such requirement.

\paragraph{Model explanation.}
\label{sec:explanation}

%First, in a real-time system on $m$ identical processors, it is known that response time bound can be guaranteed under GEDF if $U_{sum} \leq m$. Here we take the number of processors as the total resource or processor capability. However, on heterogeneous multiprocessor, the number of processors does not make sense in that processors have different speeds. Thus, we have to find our how to determine whether our processors or resource has been fully utilized.
%With heterogeneous processors we have the number of processors and the speed of processors.Thus,  naturally the processor capability or total resource should be $R_{sum}$ as defined above. Note that $R_{sum}$ is the total speed of the processors. 
In a real-time system with $m$ identical processors, it is known that response time bound can be guaranteed under GEDF if $U_{sum} \leq m$ \cite{Devi}.
For such homogeneous multiprocessor systems, the number of processors is often used to denote the total capacity. However, on a heterogeneous multiprocessor, the number of processors can no longer accurately represent the total capacity because processors have different speeds. With heterogeneous processors, we have two factors, the number of processors and the speed of each individual processor, that affect the total capacity. Thus, the total capacity of the system naturally is given by $R_{sum}$ as defined above. In other words, the total capacity is represented by the sum of the processor speeds.

Now let us consider the task model. In our model, the utilization $u_i=e_{i}/p_i$ is a quantity of speed because $e_i$ is a quantity of workload and $p_i$ is a quantity of time. In fact, using such speed to denote the utilization is intuitive because in order to meet deadlines, any task $\tau_i$ is expected to execute $e_i$ units workload within $p_i$ time units. Hence, $U_{sum}$ represents that total speed required by the task system.

%%\textbf{Example 1}. Consider a soft real-time system with four tasks:$\tau_1$=( 6(execution), 10(period) ),
%%$\tau_2$=( 6(execution), 10(period)), $\tau_3$=( 8(execution), 10(period)), $\tau_4$=( 8(execution), 10(period)). Fig 1(a) shows the schedule of these four tasks on a identical multiprocessor consisting three processors$(M_1,M_2,M_3)$ each with 1 unit speed under Global EDF. Fig 1(b) shows the schedule of these four tasks on a heterogeneous multiprocessor consisting one processor$(M_1)$ with 1 unit speed and another processor$(M_2)$ with 2 units speed. We consider the first three jobs of each task. In this example, two platforms have the same total speed. From the figures we can see at time instant 30, 6 units workload remains under the identical processor while 2 units workload remains under the heterogeneous processor.
\section{A GEDF-based Scheduling Algorithm for Heterogeneous Multiprocessor}
\label{sec:GEDF-H}
On a homogeneous multiprocessor, at any time instant, under GEDF,  when we assign $k$ ($k \leq m$) of the $n$ tasks to be executed on $k$ processors, we can arbitrarily choose processors for tasks because processors have the same speed. However, on a heterogeneous multiprocessor, if we arbitrarily choose processors for tasks, the bounded response time cannot be guaranteed as discussed in Sec.1. Motivated by this key observation, we design a GEDF-based scheduling algorithm  GEDF-H to support SRT sporadic task systems on a heterogeneous multiprocessor. GEDF-H enforces the following specific processor selection rule.

\paragraph{GEDF-H description} At any time instant  under GEDF-H, when trying to assign a job $\tau_{l,k}$(i.e., $\tau_{l,k}$ is among the $m$ highest-priority jobs at $t$) to an available processor, we consider two cases. \textbf{Case 1.} If $u_l\leq 1$, we assign $\tau_{l,k}$ to an arbitrary available processor. \textbf{Case 2.} $u_l> 1$. In this case, for some $1\leq i< z$, $\alpha_i< u_l\leq \alpha_{i+1}$. If there is an available processor $M^{'}$ in $\Psi_i$, we assign $\tau_{l,k}$ to $M^{'}$. Otherwise, by Eq. (1), there must exist at least one task $\tau_i$ with utilization $u_i \leq \alpha_i$ that has a job $\tau_{i,j}$ executing on processor $M^{'}$ in $\Psi_i$ at instant $t$. We know that, at least one processor is available at $t$ (since $\tau_{l,k}$ has not been assigned yet). Then, we move job $\tau_{i,j}$ to any available processor and assign $\tau_{l,k}$ to $M^{'}$. Note that, GEDF-H is still a job-level static-priority scheduler because we do not change a job's priority at runtime. GEDF-H gives us the following property.

\vspace{-1mm}
\begin{changemargin}{4mm}{4mm}
\textbf{(P0)} At any time instant $t$, if a job $\tau_{i,j}$ of task $\tau_i$ is executing on a processor $M^{'}$ with speed $\alpha^{'}$, we have $u_i\leq \alpha^{'}$. Let $v_i$ be the slowest speed of processors on which jobs of $\tau_i$ could execute under GEDF-H, which implies that $v_i= \alpha_{j+1}$ if $\alpha_j < u_i \leq \alpha_{j+1}$. Thus, by GEDF-H, we have
\begin{equation}
\label{eq:GEDF-H}
v_i \geq u_i
\end{equation}
\end{changemargin}
\vspace{-1mm}

\begin{figure}[t]
	\begin{center}
	\includegraphics[width=3in]{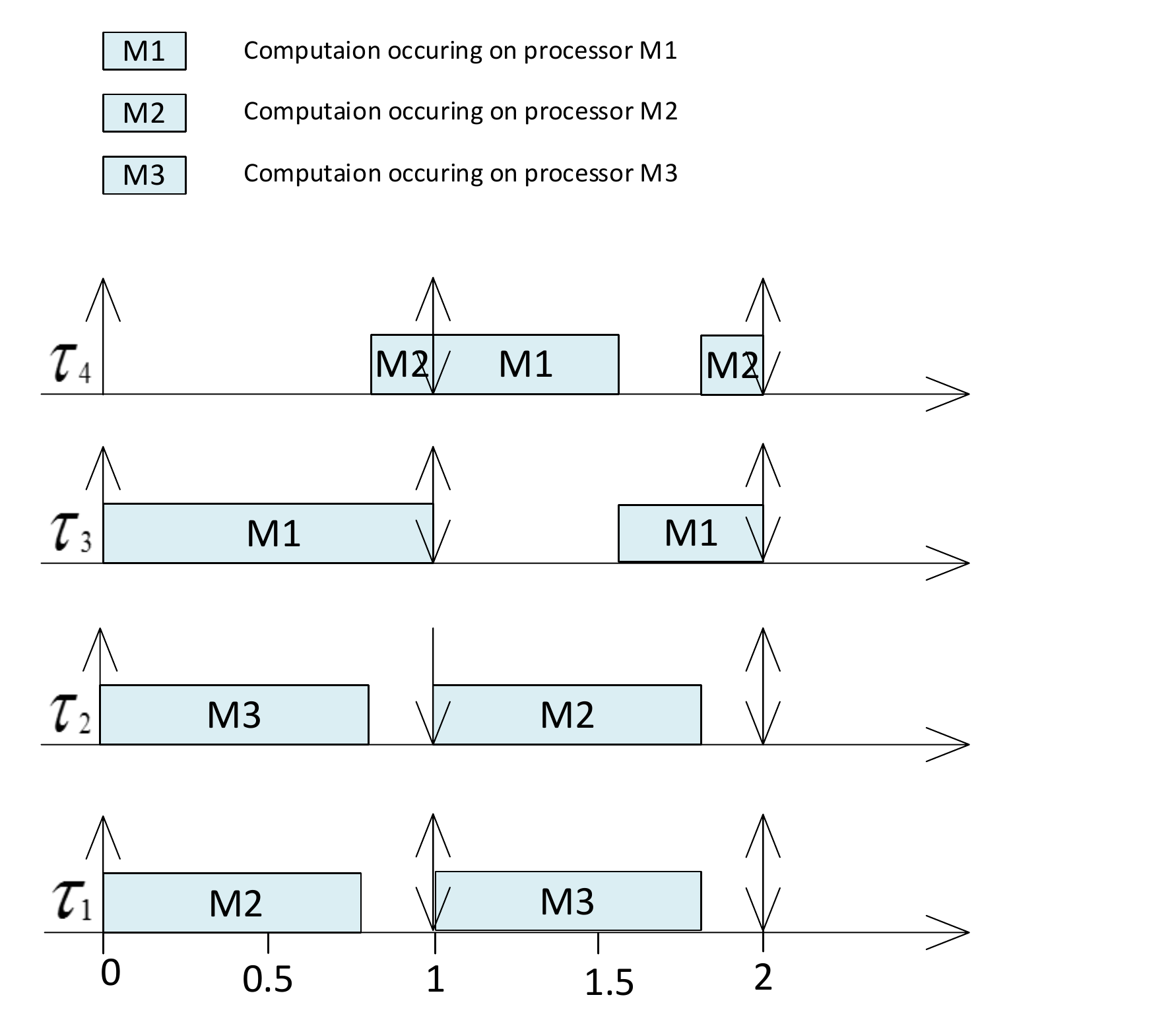} 
	\end{center} 
\vspace{-5mm}
\caption{\small GEDF-H schedule of the tasks in example 1.}
\vspace{-5mm}
\label{fig:GEDFH}
\end{figure}

Fig.\ref{fig:GEDFH} shows the GEDF-H schedule of the task system in example $1$ in time interval $[0, 2]$. At time instant 1, under GEDF-H we move $\tau_{4,1}$ from $M_2$ to $M_1$ in order to execute $\tau_{2,2}$ on $M_2$. 

Next, we derive a schedulability test for preemptive GEDF-H. For conciseness, we use GEDF-H to represent the preemptive scheduler in the following sections. Due to space constraints and the fact that the analysis for non-preemptive GEDF-H (NP-GEDF-H) is similar, we only provide a proof sketch for analyzing schedulability under NP-GEDF-H in an appendix.

\section{Schedulability Analysis for GEDF-H}
\label{sec:EDF}

We now present our preemptive GEDF-H schedulability analysis. Our analysis draws inspiration from the seminal work of Devi \cite{Devi}, and follows the same general framework. Here are the essential steps.

Let $\tau_{i,j}$ be a job of task $\tau_i$ in $\tau$, $t_d = d_{i,j}$, and $S$ be a GEDF-H schedule for $\tau$ with the following assumption.

\vspace{1.5mm}
\textbf{(A)} The response time of every job $\tau_{l,k}$, where $\tau_{l,k}$ has higher priority than $\tau_{i,j}$, is at most $x+2 \cdot  p_l$ in $S$, where $x \geq 0$.
\vspace{1.5mm}

Our objective is to find out that under which condition we could determine an $x$ such that the response time of $\tau_{i,j}$ is at most $x+2 \cdot p_i$. If we can find such $x$, by induction, this implies a response time of at most $x +2 \cdot p_l$ for all jobs of every task $\tau_l$, where $\tau_l \in \tau$. We assume that $\tau_{i,j}$ finishes after $t_d$, for otherwise, its response time is trivially equals to its period. The steps for determining the value for $x$ are as follows.
\vspace{-1.5mm}
\begin{enumerate}
\item
Determine a lower bound on the amount of work pending for tasks in $\tau$ that can compete with $\tau_{i,j}$ after $t_d$, required for the response time of $\tau_{i,j}$ to exceed $x+2 \cdot p_i$. This is dealt with in Lemma~\ref{lower_bound} in Sec.~\ref{sec:lower_bound}.
\item \vspace{-1.5mm}
Determine an upper bound on the work pending for tasks in $\tau$ that can compete with $\tau_{i,j}$ after $t_d$. This is dealt with in Lemmas~\ref{lemma:Upper_lemma1} and \ref{lemma:Upper_lemma2} in Sec.~\ref{sec:upper_bound}.
\item \vspace{-1.5mm}
Determine the smallest $x$ such that the response time of $\tau_{i,j}$ is at most $x+2 \cdot p_i$, using the above lower and upper bounds. This is dealt with in Theorem~\ref{theorem:SRTtest} in Sec.~\ref{sec:x}.
\end{enumerate} 
\vspace{-2.5mm}
\begin{definition}
\label{def:active}
A task $\tau_i$ is \textit{active} at time $t$ if there exists a job $\tau_{i,v}$ such that $r_{i,v} \leq t < d_{i,v}$.
\end{definition}

\begin{definition}
\label{def:enabled}
A job is considered to be \textit{completed} if it has finished its execution. We let $f_{i,v}$ denote the completion time of job $\tau_{i,v}$. Job $\tau_{i,v}$ is \textit{tardy} if it completes after its deadline.
\end{definition}

\begin{definition}
\label{def:pending}
 Job $\tau_{i,v}$ is \textit{pending} at time $t$ if $r_{i,v} < t < f_{i,v}$. Job $\tau_{i,v}$ is \textit{enabled} at $t$ if $r_{i,v} \leq t < f_{i,v}$, and its predecessor (if any) has completed by $t$.  %$\tau_{i,v}$ has not completed all of its execution phases by $t$. Note that $\tau_{i,v}$ is not pending at $t$ if it has completed all its execution phases by $t$ but not all of its suspension phases.
\end{definition}

%The above definitions are illustrated in Fig.~\ref{fig:readydef}.

%\begin{figure}[t] 
	%\begin{center}
       %   \includegraphics[width=3.3in]{images/readydef.pdf} 
          %\end{center}
%\vspace{-3mm}
%\caption{\small Illustration of Defs.~\ref{def:active}-\ref{def:enabled}.}
%\vspace{-3mm}
%\label{fig:readydef}
%\end{figure}
%\normalsize

\begin{definition}
\label{preemption} If an enabled job $\tau_{i,v}$ dose not execute at time $t$, then it is \textit{preempted} at $t$. 
\end{definition}

\begin{definition}
\label{def:d}
We categorize jobs based on the relationship between their priorities and those of $\tau_{i,j}$:
\begin{center}
$\textbf{d} = \lbrace {\tau_{l,v}: (d_{l,v} < t_d) \vee (d_{l,v}=t_d \wedge  l \leq i)} \rbrace$.
\end{center}
\end{definition}
Thus, \textbf{d} is the set of jobs with priority no less than that of $\tau_{i,j}$, including $\tau_{i,j}$.
\begin{definition}
\label{def:PS}
For any given sporadic task system $\tau$, a \textit{processor share} (PS) schedule is an ideal schedule where each task $\tau_i$ executes with a speed equal to $u_i$ when it is active (which ensures that each of its jobs completes exactly at its deadline). A valid PS schedule exists for $\tau$ if $U_{sum} \leq R_{sum}$ holds. 

Fig.~\ref{fig:PSschedule} shows the PS schedule of the tasks in Example 1. Note that the PS schedules on a homogeneous multiprocessor and a heterogeneous multiprocessor are identical.
\end{definition}

\begin{figure}[t]
	\begin{center}
	\includegraphics[width=3in]{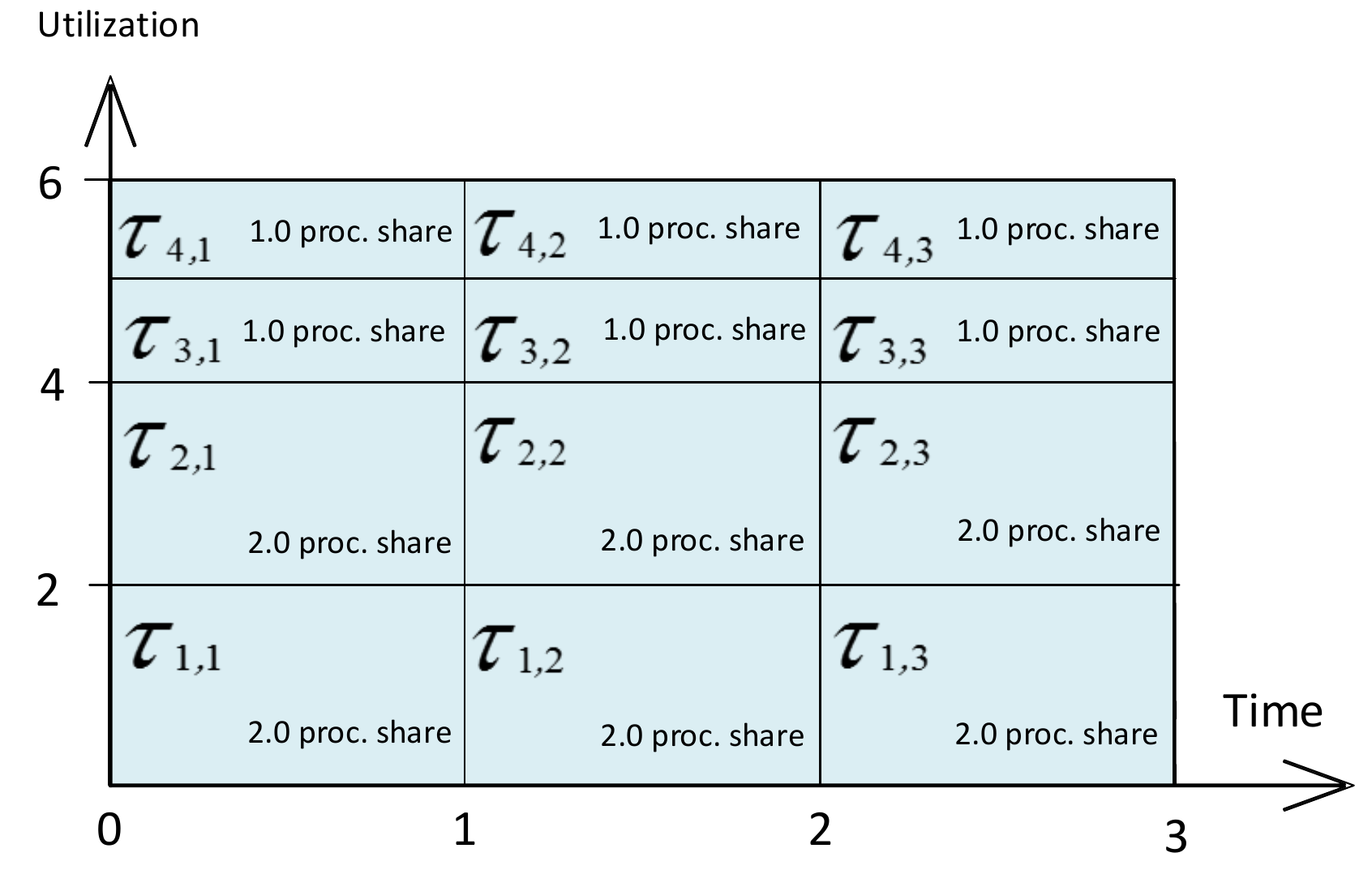} 
	\end{center} 
\vspace{-5mm}
\caption{\small PS schedule of the tasks in example 1.}
\vspace{-5mm}
\label{fig:PSschedule}
\end{figure}

By Def.~\ref{def:d}, $\tau_{i,j}$ is in $\textbf{d}$. Also jobs not in $\textbf{d}$ have lower priority than those in $\textbf{d}$ and thus do not affect the scheduling of jobs in $\textbf{d}$. For simplicity, in the rest of the paper, we only consider jobs in $\textbf{d}$ in either the GEDF-H schedule $S$ or the corresponding PS schedule.%To avoid distracting ``boundary cases,'' we also assume that the schedule being analyzed is prepended with a schedule in which no deadlines are missed that is long enough to ensure that all previously released jobs referenced in the proof exist. 

Our schedulability test is obtained by comparing the allocations to $\textbf{d}$ in the GEDF schedule $S$ and the corresponding PS schedule, both on $m$ processors, and quantifying the difference between the two. We analyze task allocations task by task. Let $A(\tau_{i,v}, t_1, t_2, S)$ denote the total workload allocation to job $\tau_{i,v}$ in $S$ in $[t_1, t_2)$. Then, the total workload done by all jobs of $\tau_i$ in $[t_1, t_2)$ in $S$ is given by \[ A(\tau_i, t_1, t_2, S)=\sum_{v \geq 1} A(\tau_{i,v}, t_1, t_2, S). \]

%In such a schedule, $\tau_i$ executes with a certain rate when it is active to make sure that each job completes exactly at its deadline. (\textit{Note that suspensions are not considered in the PS schedule}.) A valid \textit{PS} schedule exists for $\tau$ if $U_{sum} \leq m$ holds. 

Let \textit{PS} denote the PS schedule that corresponds to the GEDF-H schedule $S$ (i.e., the total allocation to any job of any task in \textit{PS} is identical to the total allocation of the job in $S$). 

The difference between the allocation to a job $\tau_{i,v}$ up to time $t$ in \textit{PS} and $S$, denoted \textit{the lag of job $\tau_{i,v}$ at time $t$ in schedule $S$}, is defined by
\begin{equation*}
lag(\tau_{i,v}, t, S) = A(\tau_{i,v}, 0, t, PS) - A(\tau_{i,v}, 0, t, S).
\end{equation*}
Similarly, the difference between the allocation to a task $\tau_{i}$ up to time $t$ in  \textit{PS} and $S$, denoted \textit{the lag of task $\tau_{i}$ at time $t$ in schedule $S$}, is defined by\small
\begin{eqnarray}
\label{eq:lag for task}
lag(\tau_i, t, S) \hspace{-1.5mm} & = & \hspace{-2mm} \sum_{v \geq 1} lag(\tau_{i,v}, t, S) \nonumber \\
                \hspace{-1.5mm} & = & \hspace{-1.5mm} \sum_{v \geq 1} \left(A(\tau_{i,v}, 0, t, PS) - A(\tau_{i,v}, 0, t, S)\right)\hspace{-0.5mm}.
\end{eqnarray}
\normalsize

The \textit{LAG} for $\textbf{d}$ at time $t$ in schedule $S$ is defined as
\begin{equation}
\label{eq:LAG for task set}
LAG(\textbf{d}, t, S) = \sum_{\tau_i: \tau_{i,v} \in \textbf{d}} lag(\tau_i, t, S).
\end{equation}
%Also, since the lag for any completed job is 0, we have 
%\begin{equation}
%\label{eq:LAGJ=sumlagT}
%LAG(J, t, S)  = \sum_{\tau_i: \tau_{i,j} \in J} lag(\tau_i, t, S).
%\end{equation}

%The total system lag of a task set $\tau$ at time $t$, denoted $LAG(\tau,t,S)$, is given by the following.
%\begin{equation}
%\label{eq:LAG for task set}
%LAG(\tau, t, S) = \sum_{\tau_{i} \in \tau} lag(\tau_{i}, t, S).
%\end{equation}

\begin{definition}
A time instant $t$ is \textit{busy} (resp. \textit{non-busy}) for a job set $J$ if there exists (resp. does not exist) an $\varepsilon > 0 $ that all $m$ processors execute jobs in $J$ during $(t,t+ \varepsilon) $. A time interval $[a,b)$ is \textit{busy} (resp. \textit{non-busy}) for $J$ if each (resp. \textit{not all}) instant within $[a,b)$ is busy for $J$. %A time instant $t$ is \textit{busy on processor $M_k$} (resp. \textit{non-busy on processor $M_k$}) for $J$ if $M_k$ executes (resp. does not execute) a job in $J$ at $t$. A time interval is \textit{busy on processor $M_k$} (resp. \textit{non-busy on processor $M_k$}) for $J$  if each instant within it is busy (resp., non-busy) on $M_k$ for $J$.
\end{definition}

The following properties follows from the definitions above.

\begin{changemargin}{4mm}{4mm}
\textbf{(P1)} If $LAG(\textbf{d}, t_2, S) > LAG(\textbf{d}, t_1, S)$, where $t_2 > t_1$, then $[t_1, t_2$) is non-busy for $\textbf{d}$. In other words, \textit{LAG} for $\textbf{d}$ can increase only throughout a non-busy interval for $\textbf{d}$ . 
\end{changemargin}

\begin{changemargin}{4mm}{4mm}
\textbf{(P2)} At any non-busy time instant $t$, at most $m-1$ tasks can have pending jobs at $t$, for otherwise $t$ would have to become busy. 
\end{changemargin}

%\vspace{-1mm}
%\begin{changemargin}{4mm}{4mm}
%\textbf{(P3)} Let $\tau_{i,j}$ be a job in task $\tau_{i}$ with deadline $d_{i,j}$. If task $\tau_{i}$ dose not have any pending job at $d_{i,j}$, then $lag(\tau_i, t, S) \leq 0$.
%\end{changemargin}
%\vspace{-1mm}
%An interval could be \textit{non-busy} for $\textbf{d}$ only if there are not enough enabled jobs in $\textbf{d}$ to occupy (i.e., execute on) all available processors.

\subsection{Lower Bound}
\label{sec:lower_bound}

Lemma~\ref{lower_bound} below provides the lower bound on $LAG(\textbf{d}, t_d, {S})$.
%Some additional notation is required to state this lemma.

\begin{lemma}
\label{lower_bound}
If $LAG(\textbf{d}, t_d, {S}) \leq R_{sum} \cdot x +p_i$ and Assumption \textbf{(A)} holds, then the response time of $\tau_{i,j}$ is at most $x + 2 \cdot p_i$,
\end{lemma}
\begin{proof}
Let $\eta_{i,j}$ be the amount of work $\tau_{i,j}$ performs by time $t_d$ in ${S}$, $0 \leq \eta_{i,j} < e_i$. Define $y$ as follows.

\begin{equation}\vspace{-1mm}
\label{eq:y}
y = x + \dfrac{\eta_{i,j}}{R_{sum}}
\end{equation}
We consider two cases.% depending on whether $[t_d,t_d+y)$ is busy. %Let $W$ be the amount of work due to jobs in $\textbf{d}$ that can compete with $\tau_{l,j}$ at or after $t_d +y$, including the work due for $\tau_{l,j}$.   

\vspace{1mm}
\textbf{Case 1.} \textit{$[t_d, t_d + y)$ is a busy interval for} $\textbf{d}$. In this case, the amount of work completed in $[t_d,t_d+y)$ is exactly $\sum_{i=1}^{p}\alpha_{i} \cdot M_{i} \cdot y = R_{sum} \cdot y$, as illustrated in Fig.4. Hence, the amount of work pending at $t_d+y$ is at most $LAG(\textbf{d}, t_d, {S}) - R_{sum} \cdot y \leq R_{sum} \cdot x + p_i - R_{sum} \cdot x - \eta_{i,j}  = p_i -\eta_{i,j}$. This remaining work will be completed(even on a slowest processor), no later than $t_d+y+p_i -\eta_{i,j} = t_d + x + \dfrac{\eta_{i,j}}{R_{sum}} +p_i -\eta_{i,j} \leq t_d+x+p_i$. Since this remaining work includes the work due for $\tau_{i,j}$, $\tau_{i,j}$ thus completes by $t_d+x+p_i$. The response time of $T_{i,j}$ is thus not more than $t_d+x+p_i-r_{i,j}= x+2 \cdot p_i$. 

\begin{figure}[t]
	\begin{center}
	\includegraphics[width=3in]{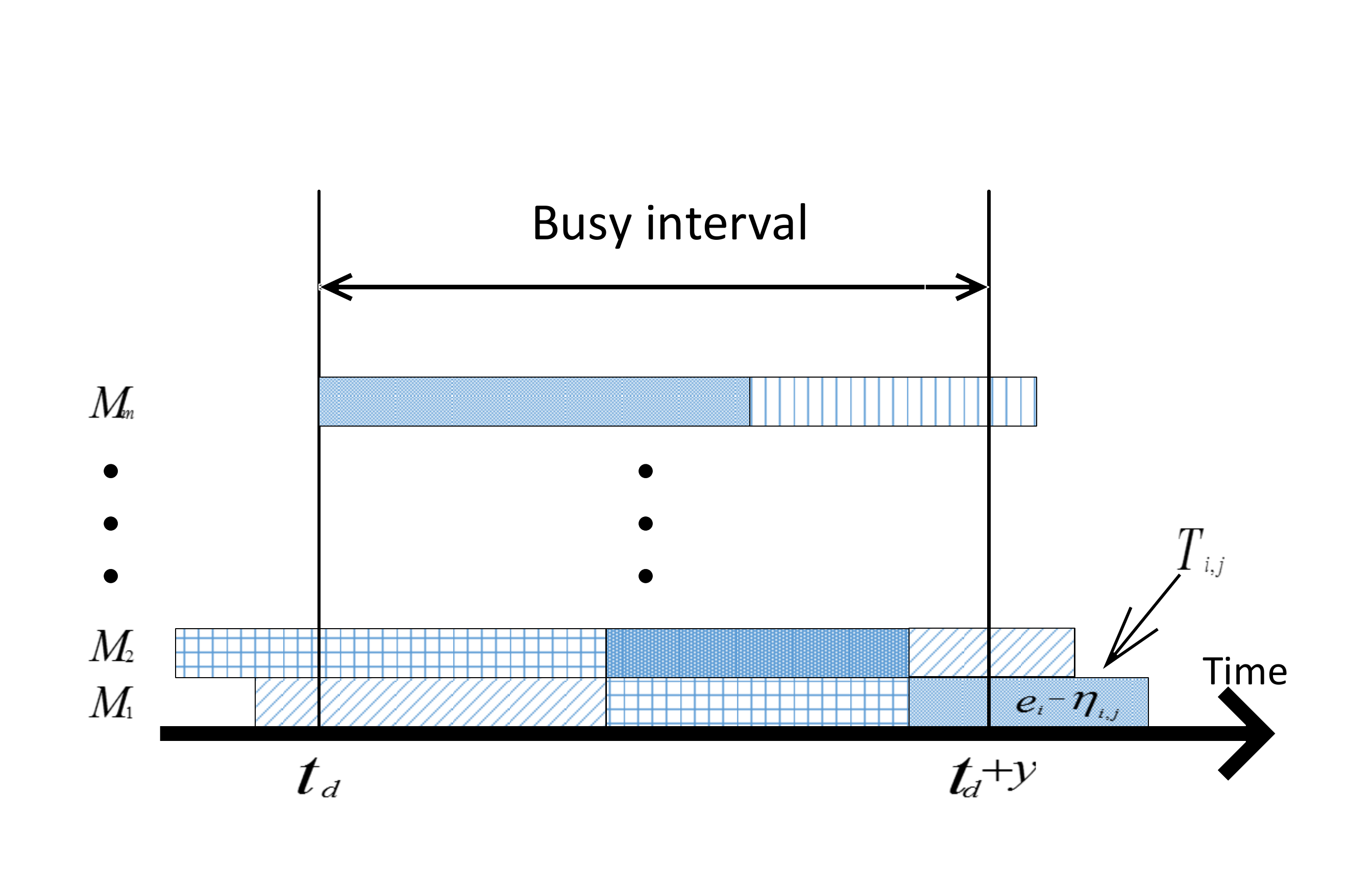} 
	\end{center} 
\vspace{-2mm}
\caption{\small $[t_d, t_d+y)$ is a busy interval.}
\vspace{-2mm}
\label{fig:busy}
\end{figure}

%This implies that no work will compete with $\tau_{l,j}$ after $t_d+y$. Since $\tau_{l,j}$ can execute and suspend for at most $e_l+s_l-\eta_l$ time units in total after $t_d + y$, we have $t_f \leq t_d+y+e_l+s_l-\eta_l=t_d+x+\dfrac{\eta_l}{m}+e_l+s_l-\eta_l\leq t_d+x+e_l+s_l$.

\vspace{1mm}
\textbf{Case 2.} \textit{$[t_d, t_d + y)$ is a non-busy interval for} $\textbf{d}$. Let $t_s$ be the earliest non-busy instant in $[t_d, t_d + y$), as illustrated in Fig.5. By Property \textbf{(P2)}, at most $m-1$ tasks can have pending jobs in $\textbf{d}$ at $t_s$. Moreover, since no jobs in $\textbf{d}$ can be released after $t_d$, we have
\vspace{-3mm}
\begin{changemargin}{8mm}{8mm}
\textbf{(P3)} At most $m-1$ tasks have pending jobs  in $\textbf{d}$ at or after $t_s$. This implies no job would be preempted at or after $t_s$.
\end{changemargin}\vspace{-3mm}

\noindent If $\tau_{i,j}$ is executing at $t_s$, then, by property \textbf{(P3)} and \textbf{(P0)}, we have

\begin{eqnarray*}
f_{i,j}  &{\leq}& t_s + \frac{e_i-\eta_{i,j}}{v_i}  \\
&&{\lbrace \rm{by}~(\ref{eq:GEDF-H}) \rbrace}\\ 
&{\leq}& t_d + y +\frac{e_i-\eta_{i,j}}{u_i} \\
&&{\lbrace \rm{by}~(\ref{eq:y}) \rbrace} \\
&{\leq}& t_d + x+  \frac{\eta_{i,j}}{R_{sum}}+\frac{e_i-\eta_{i,j}}{u_i} \\
&{\leq}& t_d + x + \frac{e_i}{u_i} \\
& = & t_d+x+p_i.
\end{eqnarray*}

%  $f_{i,j} \leq t_s + \frac{e_i-\eta_{i,j}}{v_i} \stackrel{\lbrace \rm{by}~(\ref{eq:GEDF-H}) \rbrace}{\leq}  t_d + y +\frac{e_i-\eta_{i,j}}{u_i} \stackrel{\lbrace \rm{by}~(\ref{eq:y}) \rbrace}{\leq} t_d + x+  \frac{\eta_{i,j}}{R_{sum}}+\frac{e_i-\eta_{i,j}}{u_i} \leq t_d + x + \frac{e_i}{u_i} = t_d+x+p_i$. 
Thus, the response time of $T_{i,j}$ is not more than $f_{i,j}-r_{i,j}= f_{i,j}-t_d+p_i \leq x+ 2 \cdot p_i$.

Else, $\tau_{i,j}$ is not executing at $t_s$ and $\eta_{i,j}= 0$, which means the predecessor job $\tau_{i,j-1}$ has not completed by $t_s$. Because $d_{i,j-1}$ = $t_d-p_{i}$, by Assumption \textbf{(A)}, $f_{i,j-1} \leq r_{i,j-1} + x+ 2 \cdot p_i =d_{i,j-1}-p_i+x+2 \cdot p_i = t_d+x$. Thus, combined with property \textbf{(P3)} and \textbf{(P0)}, $f_{i,j}\leq f_{i,j-1}+\frac{e_i}{v_i}\leq t_{d}+x+\frac{e_i}{v_i} 
\stackrel{\lbrace \rm{by}~(\ref{eq:GEDF-H}) \rbrace}{\leq}
t_{d}+x+\frac{e_i}{u_i}= t_d+x+p_i$. The response time of $\tau_{i,j}$ is thus not more than $x+ 2 \cdot p_i$. \qedhere

\begin{figure}[t]
	\begin{center}
	\includegraphics[width=3in]{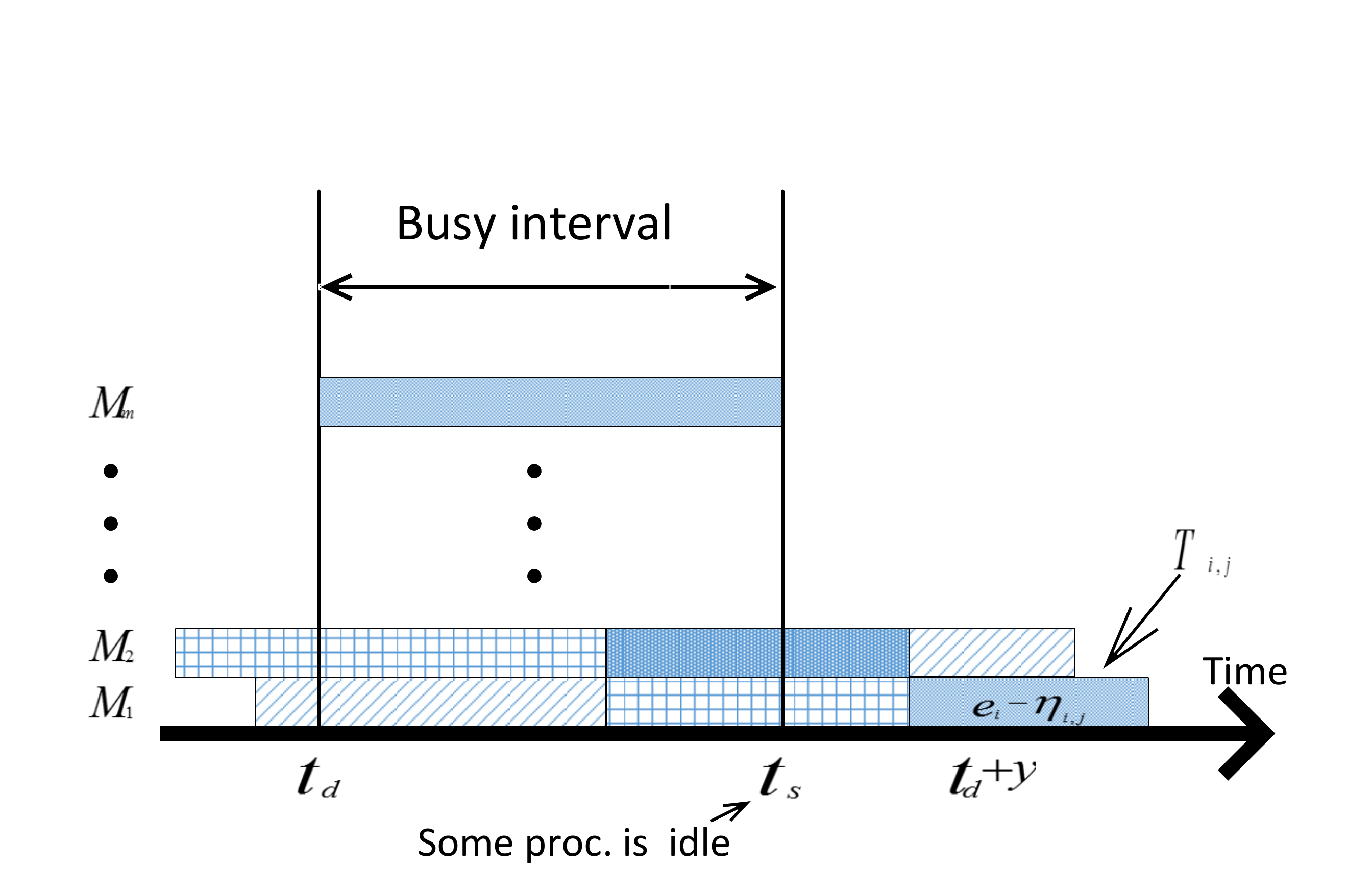} 
	\end{center} 
\vspace{-2mm}
\caption{\small $[t_d, t_d+y)$ is a non-busy interval.}
\vspace{-2mm}
\label{fig:nonbusy}
\end{figure}

\end{proof}

\subsection{Upper Bound}
\label{sec:upper_bound}

In this section, we determine an upper bound on $LAG(\textbf{d}, t_d, {S})$.

\begin{definition}
Let $t_n \leq t_d$ be the latest non-busy instant by $t_d$ for \textbf{d}, if any; otherwise, $t_n = 0$.
\end{definition}

By the above definition and Property \textbf{(P1)}, we have 
\begin{equation}
\label{eq:LAGincrease}
LAG(\textbf{d}, t_d, {S}) \leq LAG(\textbf{d}, t_n, {S}).
\end{equation}

\begin{lemma}
\label{lemma:Upper_lemma1}
For any task $\tau_i$, if $\tau_i$ has pending jobs at $t_{n}$ in the schedule ${S}$, then we have
\small
\[
 lag(\tau_i, t_n, {S}) \leq
  \begin{cases}
  e_i  &  \hspace{0mm} \hspace{-2.5mm} \text{if } d_{i,k} \geq t_n \\
   u_i \cdot x + e_i+u_i \cdot (p_i- \frac{e_i}{\alpha_{max}})   & \hspace{0mm} \hspace{-2.5mm} \text{if } d_{i,k} < t_{n}
  \end{cases}
\] \normalsize
where $d_{i,k}$ is the deadline of the earliest released pending job of $\tau_i$, $\tau_{i,k}$, at time $t_{n}$ in ${S}$. 
\end{lemma}

\begin{proof}
Let $\gamma_{i,k}(\gamma_{i,k} < e_i)$be the amount of work $\tau_{i,k}$ performs before $t_{n}$.

By the selection of $\tau_{i,k}$, we have $lag(\tau_i, t_n, {S}) = \sum_{h \geq k} lag(\tau_{i,h}, t_n, {S})  = \sum_{h \geq k} \big(A(\tau_{i,h}, 0, t_n,PS) - A(\tau_{i,h}, 0, t_n, {S})\big)$. By the definition, $A(\tau_{i,h}, 0, t_n, {S}) = A(\tau_{i,h}, r_{i,h}, t_n, {S})$. Thus,
\begin{eqnarray}
\label{summinglemma2}
&& lag(\tau_i, t_n, {S}) \nonumber \\
 &= & A(\tau_{i,k}, r_{i,k}, t_n, PS) - A(\tau_{i,k}, r_{i,k}, t_n, {S}) \nonumber \\
& & + \sum_{h >k} \big(A(\tau_{i,h}, r_{i,h}, t_n, PS) \nonumber \\
& & - A(\tau_{i,h}, r_{i,h}, t_n, {S}) \big).
\end{eqnarray}

By the definition of $PS$, $A(\tau_{i,k}, r_{i,k},t_n, PS) \leq e_i$, and $\sum_{h > k} A(\tau_{i,h}, r_{i,h}, t_n, PS)\leq {u_i} \cdot \max(0,t_n - d_{i,k})$. By the selection of $\tau_{i,k}$, $A(\tau_{i,k}, r_{i,k}, t_n, {S}) = \gamma_{i,k}$, and $\sum_{h > k} A(\tau_{i,h}, r_{i,h}, t_n, {S}) = 0$. By setting these values into (\ref{summinglemma2}), we have
\begin{equation}
\label{lemma3case1-}
lag(\tau_i, t_n, {S}) \leq e_i - \gamma_{i,k} + u_i \cdot \max(0,t_n - d_{i,k}).
\end{equation}

There are two cases to consider.

\vspace{1mm}
\textbf{Case 1.} $d_{i,k} \geq t_n$.
In this case, (\ref{lemma3case1-}) implies $lag(\tau_i, t_n, {S}) \leq e_i - \gamma_{i,k} \leq e_i$.

\vspace{2mm}
\textbf{Case 2.} $d_{i,k} < t_n$. In this case, because $t_n \leq t_d$ and $d_{l,j}=t_d$, $\tau_{i,k}$ is not the job $\tau_{l,j}$. Thus, by Assumption \textbf{(A)}, $\tau_{i,k}$ has a response time of at most $x + 2 \cdot p_i$. Since $\tau_{i,k}$ is the earliest pending job of $\tau_i$ at time $t_n$, the earliest possible completion time of $\tau_{i,k}$ is at $t_n+\frac{e_i-\gamma_{i,k}}{\alpha_{z}}$ (executed on the fastest processor). Thus, we have $t_n+\frac{e_i-\gamma_{i,k}}{\alpha_{z}} \leq r_{i,k}+x+2 \cdot p_i = d_{i,k} +x + p_i $, which gives $t_n - d_{i,k} \leq x + \frac{\gamma_{i,k}}{\alpha_{z}}+p_i-\frac{e_i}{\alpha_{z}}$. Setting this value into (\ref{lemma3case1-}), we have $lag(\tau_i, t_n, {S}) \leq e_i-\gamma_{i,k}+u_i \cdot (x + \frac{\gamma_{i,k}}{\alpha_{z}}+p_i-\frac{e_i}{\alpha_{z}}) \leq u_i \cdot x + e_i+ u_i \cdot (p_i- \frac{e_i}{\alpha_{z}}) $.
\end{proof}

\begin{definition}
\label{def:UE}
Let $\overline{U}_{m-1}$ be the sum of the $m-1$ largest $u_i$ values among tasks in $\tau$. Let $\overline{E}$ be the largest value of the expression $ \sum_{\tau_i \in \psi,\alpha_{j} \in \varphi}  (e_i+u_i \cdot (p_i- \frac{e_i}{\alpha_{j}}))$, where $\psi$ denotes any set of $m-1$ tasks in $\tau$ and $\varphi$ denotes the set of speeds of $m-1$ processors that are the $m-1$ fastest processors in the system.
\end{definition}

Lemma~\ref{lemma:Upper_lemma2} below upper bounds $LAG(\textbf{d}, t_d, {S})$.

\begin{lemma}
\label{lemma:Upper_lemma2}
With the Assumption \textbf{(A)}, $LAG(\textbf{d}, t_d, {S}) \leq \overline{U}_{m-1} \cdot x + \overline{E}$.
\end{lemma}
\begin{proof}

By (\ref{eq:LAGincrease}), we have $LAG(\textbf{d}, t_d, {S}) \leq LAG(\textbf{d}, t_n, {S})$. By summing individual task lags at $t_n$, we can bound $LAG(\textbf{d}, t_n, {S})$. If $t_n=0$, then $LAG(\textbf{d}, t_n, {S})=0$, so assume $t_n > 0$.

Given that the instant $t_n$ is non-busy, by Property \textbf{(P2)}, at most $m-1$ tasks can have pending jobs at $t_n$. Let $\theta$ denote the set of such tasks. Therefore, by Eq.~(\ref{eq:LAGincrease}), we have
\begin{eqnarray*}
\label{lemma3}
LAG(\textbf{d}, t_d, {S})  &{\leq}& LAG(\textbf{d}, t_n, {S})  \\
&&{\lbrace \rm{by}~Eq.~(\ref{eq:LAG for task set}) \rbrace}\\ 
&{=}& \sum_{\tau_i: \tau_{i,v} \in \textbf{d}} lag(\tau_i, t_n, {S}) \nonumber \\
&&{\lbrace \rm{by}~\rm{Lemma}~\ref{lemma:Upper_lemma1} \rbrace} \\
&{\leq}& \sum_{\tau_i \in \theta} \left( u_i \cdot x+e_i+u_i \cdot \left(p_i- \frac{e_i}{\alpha_{z}} \right) \right).
\end{eqnarray*}

%$ LAG(\textbf{d}, t_d, {S}) \stackrel{\lbrace \rm{by}~(\ref{eq:LAGincrease}) \rbrace}{\leq}  LAG(\textbf{d}, t_n, {S}) \stackrel{\lbrace \rm{by}~(\ref{eq:LAG for task set}) \rbrace}{=}   \sum_{\tau_i: \tau_{i,v} \in \textbf{d}} lag(\tau_i, t_n, {S}) \stackrel{\lbrace \rm{by}~\rm{Lemma}~\ref{lemma:Upper_lemma1} \rbrace}{\leq} \sum_{\tau_i \in \theta} ( u_i \cdot x+e_i+u_i \cdot (p_i- \frac{e_i}{\alpha_{z}}))$. 
Since two jobs cannot be executed on the same processor at any time instant, $ LAG(\textbf{d}, t_d, {S})$ reaches its maximal value when the  $m-1$ tasks in $\theta$ execute on the $m-1$ fastest processors. Thus, \\
\begin{eqnarray*}
&& LAG(\textbf{d}, t_d, {S}) \\ 
 &\leq\ & \sum_{\tau_i \in \psi,\alpha_{j} \in \varphi} \left( u_i \cdot x+e_i+u_i \cdot \left(p_i- \frac{e_i}{\alpha_{z}} \right) \right) \\ 
&&{\lbrace \rm{by}~\rm{Def.}~\ref{def:UE} \rbrace}\\
&{\leq}& \overline{U}_{m-1} \cdot x + \overline{E}. \hspace{37mm} \qedhere
\end{eqnarray*}
\end{proof}
%that have  enabled jobs at $t_n-1$ with deadlines at or before $t_n-1$. Any task not in $\theta$ can have an enabled job suspending at $t_n-1$, allowing $t_n-1$ to be non-busy. Note that if task $\tau_i$ does not have an enabled job at $t_n-1$, then $lag(\tau_i, t_n, {S}) \leq 0$. 

%\begin{align*}
%\label{eq:lemma4}
%LAG(\textbf{d}, t_d, {S}) &\leq LAG(\textbf{d}, t_n, {S})\\%
%& \mathrm{\lbrace by \mbox{ } (\ref{eq:LAGJ=sumlagT}) \rbrace}  \\
%&= \sum_{\tau_i: \tau_{i,v} \in \textbf{d}} lag(\tau_i, t_n, {S})  \\
%& \leq \sum_{\tau_i \in \theta} lag(\tau_i, t_n, {S}) + \sum_{\tau_j \in \tau - \theta} lag(\tau_j, t_n, {S})  \\
%& \mathrm{\lbrace by \mbox{ } Lemma \mbox{ } \ref{lemma:Upper_lemma1} \rbrace}  \\
%& \leq \sum_{\tau_i \in \theta} (\overline{u_i} \cdot x + e_i +s_i+\overline{u_i} \cdot s_i ) \\
%&\hspace{4mm} + \sum_{\tau_j \in \tau-\theta} (e_j+s_j) \\
%& \mathrm{\lbrace by \mbox{ } Def. \mbox{ } \ref{def:UE} \rbrace}  \\
%& \leq \overline{U}_{m-1} \cdot x + \overline{E}. \qedhere
%\end{align*}

\subsection{Determining $x$}
\label{sec:x}

Setting the upper bound on $LAG(\textbf{d}, t_d, {S})$ in Lemma~\ref{lemma:Upper_lemma2} to be at most the lower bound in Lemma~\ref{lower_bound} will ensure that the response time of $\tau_{i,j}$ is at most $x + p_i $. The resulting inequality can be used to determine a value for $x$. By Lemmas~\ref{lower_bound} and \ref{lemma:Upper_lemma2}, this inequality is $R_{sum} \cdot x +p_i \geq \overline{U}_{m-1} \cdot x + \overline{E}$. Solving for $x$, to make a $x$ valid for all tasks, we have 
\begin{eqnarray}
\label{xrange}\vspace{-1mm}
 x \geq \dfrac{\overline{E}-p_{min}}{R_{sum}-\overline{U}_{m-1}}. 
\end{eqnarray} \vspace{-1mm}
By $U_{sum} \leq R_{sum}$ and Defs.\ref{def:UE}, $\overline{U}_{m-1} < R_{sum}$ clearly holds. Let
\begin{eqnarray}
\label{xvalue}\vspace{-1mm}
 x =max( 0 , \dfrac{\overline{E}-p_{min}}{R_{sum}-\overline{U}_{m-1}}) ,
\end{eqnarray} \vspace{-1mm} 
then the response time of $\tau_{i,j}$ will not exceed $x+2 \cdot p_i$ in ${S}$.

By the above discussion, the theorem below follows.

\begin{theorem}
\label{theorem:SRTtest}
With $x$ as defined in (\ref{xvalue}), the response time of any task $\tau_i$ scheduled under GEDF-H is at most $x + 2 \cdot p_i $, provided $U_{sum} \leq R_{sum}$.
\end{theorem}

\section{Experiment}
\label{sec:Experiment}

Although GEDF-H ensures SRT schedulability with no utilization loss, the magnitude of the resulting response time bound is also important. In this section, we describe experiments conducted using randomly-generated task sets to evaluate the applicability of the response time bound given in Theorem~\ref{theorem:SRTtest}. Our goal is to examine how large the magnitude of response time is.

\paragraph{Experimental setup.} We simulate the Intel's QuickIA heterogeneous prototype platform \cite{chitlur2012quickia} in our experiments. The QuickIA platform contains two kinds of processors and each kind contains two processors. We assume that two of the processors $M_1$ and $M_2$ have unit speed and the other two processors $M_3$ and $M_4$ have two-unit speed, i.e., $\alpha_1 =1$ and $\alpha_2 = 2$. The unit time is assumed to be $1 ms$.

By the definitions of $\Psi$ and $\varPhi$, we have $\Psi_0= \{ M_1, M_2, M_3, M_4 \}$, $|\Psi_0| = 4$, $\Psi_1= \{M_3, M_4\}$, $|\Psi_1| = 2$. We generated tasks as follows. Task periods were uniformly distributed over $[10ms, 600ms]$. First, we generated tasks in $\varPhi_1$. According to Eq.~\ref{eq:Restriction},  $|\varPhi_1| \leq |\psi_1| = 2$ and the utilization of tasks in $\varPhi_1$ is at most $2$. We thus first randomly generated the number of tasks in $\varPhi_1$ from $0$ to $2$, and task utilizations were generated using the uniform distribution $(1, 2]$. Task execution costs were calculated from periods and utilizations. Then, we generated tasks in $\varPhi_0/\varPhi_1$. The utilization of tasks in $\varPhi_0/\varPhi_1$ is not more than $1$. These task utilizations were generated using three uniform distributions$: [0.001, 0.05]$(light), $[0.05, 0.2]$(medium) and $[0.2, 0.5]$(heavy). For each experiment, 10,000 task sets were generated. Each such task set was generated by creating tasks until total utilization exceeded $R_{sum}=6$, and by then reducing the last task's utilization so that the total utilization equaled $R_{sum}$.

\begin{figure*}[!ht]
  \centering
  \subfloat[Heavy task utilization]{\label{fig:2}\includegraphics[width=0.47\textwidth]{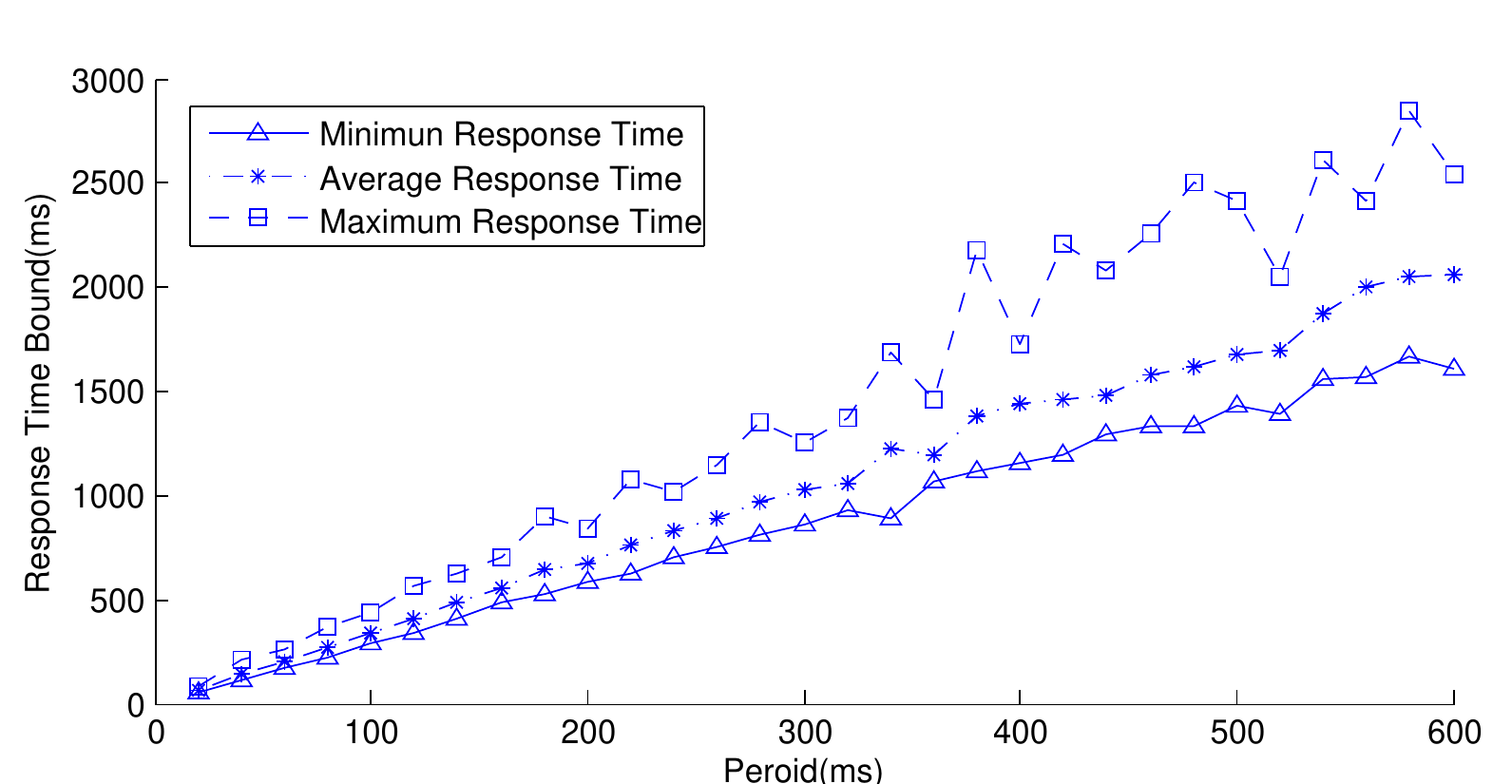}} \hspace{3mm}
  \subfloat[Period = 100$ms$]{\label{fig:3}\includegraphics[width=0.47\textwidth]{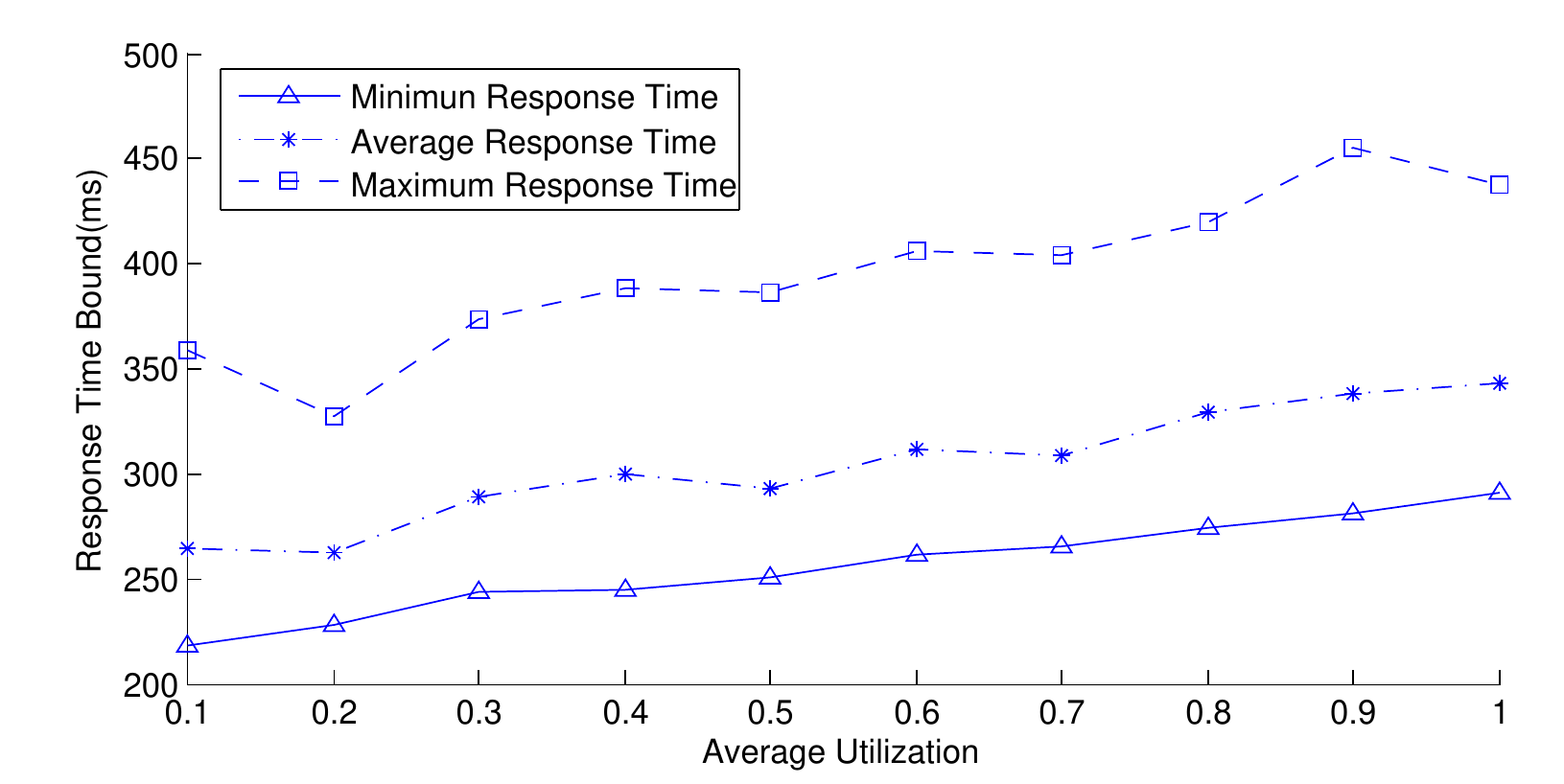}}\\
  \vspace{-2mm}
  \subfloat[Medium task utilization]{\label{fig:4}\includegraphics[width=0.47\textwidth]{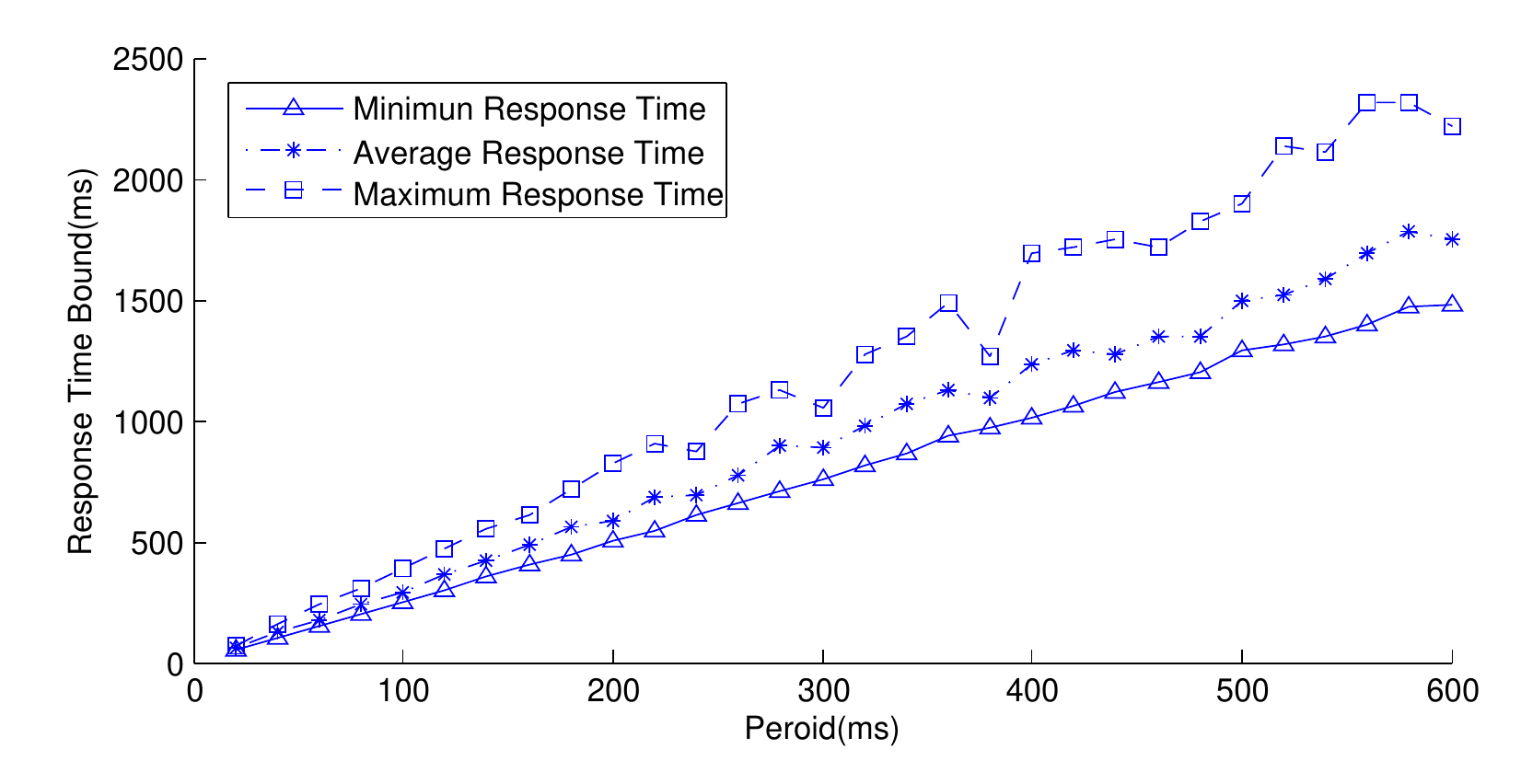}}  \hspace{3mm}           
  \subfloat[Period = 300$ms$]{\label{fig:6}\includegraphics[width=0.47\textwidth]{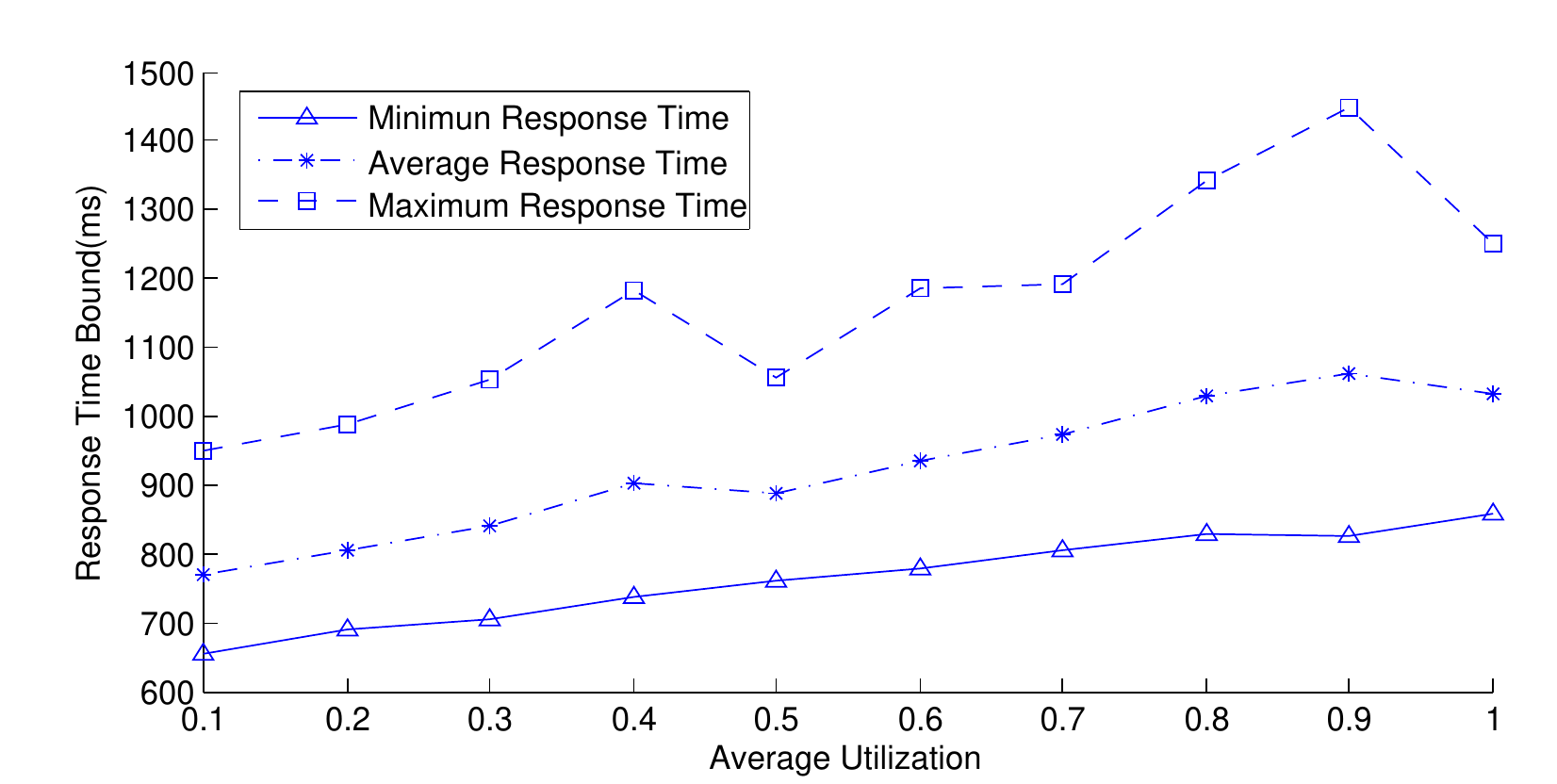}}\\
  \vspace{-2mm}  \subfloat[Light task utilization]{\label{fig:7}\includegraphics[width=0.47\textwidth]{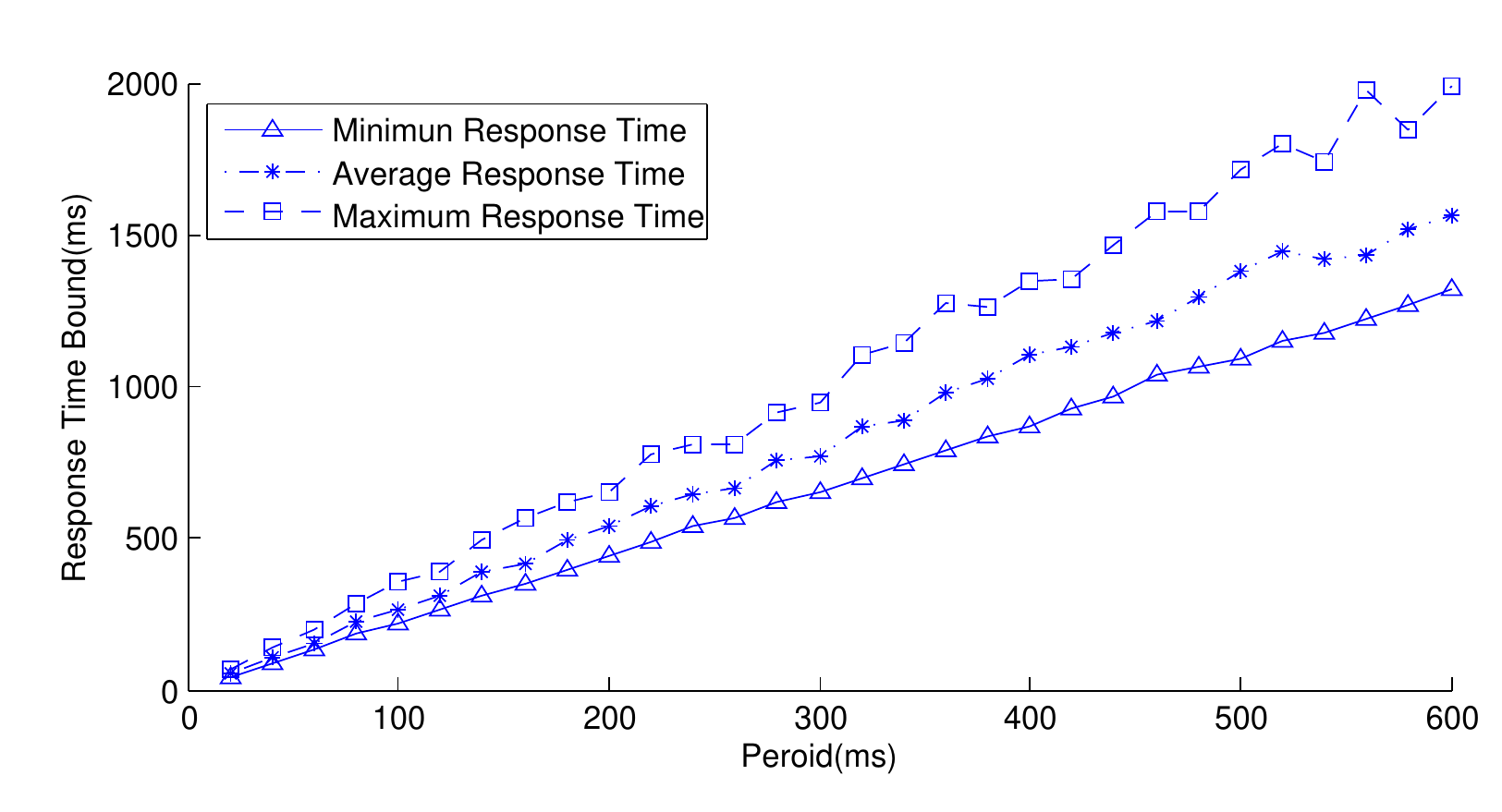}} \hspace{3mm}          
  \subfloat[Period = 600$ms$]{\label{fig:9}\includegraphics[width=0.47\textwidth]{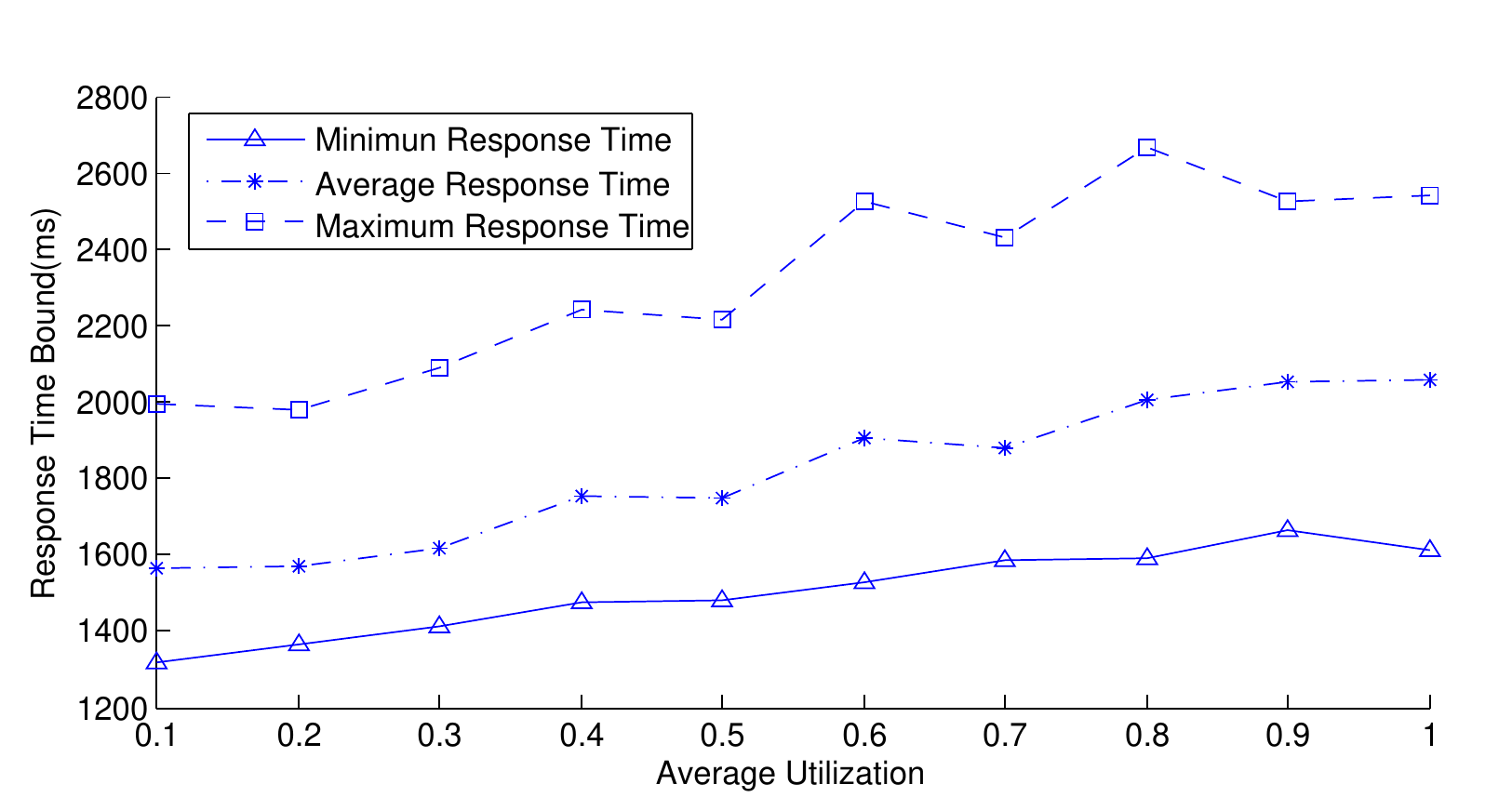}}
  \caption{\small Response time bounds. In all six graphs, the $y$-axis denotes the response time bound value. Each graph gives three curves plotting the maximum, average, and minimum response time bound among tasks, respectively. In the first column of graphs, the $x$-axis denotes the task periods. Light, medium, and heavy task utilizations are assumed in insets (a), (c), and (e), respectively. In the second column of graphs, the $x$-axis denotes the average task utilization of the generated task system. Three specific period values, 100ms, 300ms, and 600ms, are assumed in insets (b), (d), and (f), respectively. Note that the average task utilization is at most 1 in these experiments. This is because according to our task generation strategy, the number of tasks with utilization no greater than 1 is much larger than the number of tasks with utilizations greater than 1.
  } \vspace{-2mm} \normalsize
  \label{fig:exp}
\end{figure*}

\paragraph{Results.} The obtained results are shown in Fig.~\ref{fig:exp} (the organization of which is explained in the figure's caption). Each graph in Fig.~\ref{fig:exp} contains three curses, which plots the calculated maximum response time bound, average response time bound, and minimum response time bound among all tasks in the system, respectively. As seen in Figs.\ref{fig:exp}(a), (c), and (e), in all tested scenarios, the maximum response time bound is smaller than five task periods, while the average response time bound is slightly larger than three task periods (but smaller than four task periods). One observation herein is that when task utilizations become heavier, the response time bounds increase. This is intuitive because the denominator of Eq.~(\ref{xvalue}) becomes smaller when task utilizations are heavier. Moreover, as seen in Figs.~\ref{fig:exp}(b), (d), and (f), the response time bounds under GEDF-H  slightly increase along with the increase of the average task utilization of the system, under three fixed task period scenarios. Under these scenarios, the maximum response time bound is within three task periods and the average response time bound is within two task periods. To conclude, GEDF-H not only guarantees SRT schedulability with no utilization loss, but can provide such a guarantee with low predicted response time.

\section{Conclusion}
\label{sec:Conclusion}
We have shown that SRT sporadic task systems can be supported under GEDF-H on a heterogeneous multiprocessor with no utilization loss provided bounded response time is acceptable. GEDF-H is identical to GEDF except that it enforces a specific processor selection rule. As demonstrated by experiments presented herein, GEDF-H is able to guarantee schedulability with no utilization loss while providing low predicted response time. For the future work, we plan to design better algorithm that can reduce the job migration cost. Compared to GEDF, GEDF-H may incur more job migrations among processors due to the specific processor selection rule. Also it would be interesting to extent this work to hard-real systems and self-suspending task systems.

\begin{spacing}{0.1}
\bibliographystyle{plain} 
%\small
%\footnotesize
\bibliography{SRTHet}
\end{spacing}

\vspace{-1mm}
\section*{Appendix: Schedulability Analysis for NP-GEDF-H}
\label{sec:NP-EDF}

We now present our non-preemptive GEDF-H (NP-GEDF-H) schedulability analysis. Due to space constrains, we only provide the sketch of the proof.

\begin{definition}
\label{def:block job*}
For any time instant $t$, if there exists an $\varepsilon > 0$ such that during interval $[t,t+\varepsilon)$ there is an enabled job $\tau_{i,j}$ in \textbf{d} is not executing while any job $\tau_{k,l}$ not in \textbf{d} is executing on some processor during this interval, we say $\tau_{i,j}$ is blocked by $\tau_{k,l}$ at time $t$. $\tau_{i,j}$ is a blocked job; $\tau_{k,l}$ is a blocking job. $t$ is a blocking instant.
\end{definition}

\begin{definition}
\label{def:block interval*}
An interval $[a,b)$ is a blocking interval if every instant in it is a blocking instant. A blocking interval is said to be a maximal blocking interval if for any $c < a$, $[c,b)$ cannot be a blocking interval.
\end{definition}
 
\begin{definition}
\label{def:BIG B*}
Let $\beta$ denote the set of jobs not in \textbf{d} that block one or more jobs in \textbf{d} at some instants before $t_d$ and may continue to execute at $t_d$ under NP-GEDF-H. Let $B(\beta, t_d, S^{*})$ denote the total workload pending for jobs in $\beta$ at $t_d$. 
\end{definition}

\paragraph{Response time bound under NP-GEDF-H.} In the analysis of GEDF-H scheduling, only the workload pending for jobs in \textbf{d} can compete with $\tau_{i,j}$. However, under NP-GEDF-H, jobs not in \textbf{d} are still able to compete with $\tau_{i,j}$. Even though such jobs have lower priority, they cannot be preempted once they start execution before $t_d$. Hence, the pending workload from blocking jobs should be taken into consideration. After accurately defining the pending work, we are able to follow the similar analysis for NP-GEDF-H. We make the following similar assumption.

\textbf{(A-NP)} The response time of every job $\tau_{l,k}$, where $\tau_{l,k}$ has higher priority than $\tau_{i,j}$, is at most $x+2 \cdot p_l$ in $S$, where $x \geq 0$.

By the discussion above, the total pending work is presented by 
\begin{eqnarray*}
\label{NPLAG}\vspace{-1mm}
&LAG(\textbf{d}, t_d, {NP-GEDF-H})+B(\beta, t_d, S^{*}).
\end{eqnarray*} \vspace{-1mm}
To derive the lower bound of $LAG(\textbf{d}, t_d, {NP-GEDF-H})+B(\beta, t_d, S^{*})$, we have following parallel Lemma \ref{lemma:NP-Lower_lemma} for NP-EDFH. The proof is the same to the proof of Lemma \ref{lower_bound}

\begin{lemma}
\label{lemma:NP-Lower_lemma}
If $LAG(\textbf{d}, t_d, {NP-GEDF-H})+B(\beta, t_d, S^{*}) \leq R_{sum} \cdot x +2 \cdot p_i$ and the Assumption \textbf{(A-NP)} holds, then the response time of $\tau_{i,j}$ is at most $x + 2 \cdot p_i$,
\end{lemma}

To derive the upper bound of $LAG(\textbf{d}, t_d, {NP-GEDF-H})+B(\beta, t_d, S^{*})$, we have the following parallel Lemma \ref{lemma:NP-upper_lemma} for NP-GEDF-H. The proof is slightly different from the proof of Lemma \ref{lemma:Upper_lemma2}. Let $E^{*}$ be the largest value of the expression $ \sum_{\tau_i \in \psi ,\alpha_{j} \in \varphi}  (e_i+u_i \cdot e_i \cdot (1-\frac{1}{\alpha_j})) +e_k$, where $\psi$ denotes any set of $m-1$ tasks in $\tau$, $\varphi$ denotes the set of speed of $m-1$ processors those are the most $m-1$ fastest, $e_k$ is the execution of any $\tau_k$ not in $\psi $.

\begin{lemma}
\label{lemma:NP-upper_lemma}
With Assumption \textbf{(A-NP)}, $LAG(\textbf{d}, t_d, {NP-GEDF-H})+B(\beta, t_d, S^{*}) \leq \overline{U}_{m-1} \cdot x + E^{*}$. 
\end{lemma}
\begin{proof}
Let $t_n$ be the latest non-busy instant before $t_d$. For NP-GEDF-H, we consider following two cases. \textbf{Case 1.} $t_n$ is not a blocking instant, we are able to do the analysis similar to Lemma 2. \textbf{Case 2.} $t_n$ is a blocking instant. Let $[t^{'},t_n)$ be the maximal blocking interval. And we first derive the upper bound for $LAG(\textbf{d}, t^{'}, {NP-GEDF-H})+B(\beta, t^{'}, S^{*})$; then extend it to $LAG(\textbf{d}, t_n, {NP-GEDF-H})+B(\beta, t_n, S^{*})$.
\end{proof}

\vspace{-1mm}
By the Lemma \ref{lemma:NP-Lower_lemma} and \ref{lemma:NP-upper_lemma}, the theorem below immediately follows.
\vspace{-1mm}

\begin{theorem}
\label{theorem:NPSRTtest}
With $
 x= max( 0 , \dfrac{E^{*}-p_{min}}{R_{sum}-\overline{U}_{m-1}}) $, the response time of any task $\tau_i$ scheduled under NP-GEDF-H is at most $x + 2 \cdot p_i $, provided $U_{sum} \leq R_{sum}$.
\end{theorem}

\end{document}